\documentclass[12pt]{article}
\usepackage[left=2cm,right=2cm, top=2.4cm,bottom=2.4cm,bindingoffset=0cm]{geometry}
\usepackage{graphicx}
\usepackage{amsmath}
\usepackage{amsthm}
\usepackage{amssymb}
\usepackage{mathtools}
\usepackage{xfrac}
\usepackage{faktor}
\usepackage{setspace}
\newtheorem{theorem}{Theorem}[section]

\newtheoremstyle{italics}{}{}{\itshape}{}{\bfseries}{ }{ }{}
\theoremstyle{italics}
\newtheorem{prop}[theorem]{Proposition}

\newtheorem{lemma}[theorem]{Lemma}
\newtheorem{cor}[theorem]{Corollary}

\newtheoremstyle{noitalics}{}{}{}{}{\bfseries}{ }{ }{}
\theoremstyle{noitalics}
\newtheorem{definition}[theorem]{Definition}
\newtheorem{remark}[theorem]{Remark}
\newtheorem{example}[theorem]{Example}
\newtheorem*{example*}{Example}

\newcommand{\F}{\mathbb{F}}
\newcommand{\R}{\mathbb{R}}
\newcommand{\Z}{\mathbb{Z}}
\newcommand{\C}{\mathbb{C}}
\newcommand{\N}{\mathbb{N}}

\newcommand{\Mat}{\mathrm{M}}
\newcommand{\Hom}{\mathrm{Hom}}
\newcommand{\Ker}{\mathrm{Ker}\,}
\newcommand{\rank}{\text{rank}}
\newcommand{\End}{\mathrm{End}}

\newcommand{\Gr}{\mathrm{Gr}}

\newcommand{\Proj}{\mathbb{P}}
\newcommand{\PEC}{\mathcal P}
\newcommand{\D}{\mathcal{D}}
\newcommand{\Q}{\mathcal{Q}}
\newcommand{\Rem}{\mathcal{R}}
\newcommand{\wD}{\widehat{\D}}
\newcommand{\wQ}{\widehat{\Q}}

\newcommand{\Tr}{\mathrm{Tr}}

\newcommand{\K}{\Ker \,\D}
\newcommand{\spanof}{\text{span}}

\newcommand{\Points}{\Pi}
\newcommand{\Lines}{\Lambda}
\newcommand{\Inc}{\mathrm I}
\usepackage{hyperref}

\newcommand{\DN}{\mathbb D}

\newcommand{\AL}{\mathcal A}
\newcommand{\RP}{\R\Proj}
\usepackage{xcolor}

\newcommand{\eps}{\varepsilon}
\usepackage{tikz, titlesec, tikz-cd}
\usetikzlibrary{calc,intersections,through,backgrounds}
\usepackage{mathtools}
\mathtoolsset{showonlyrefs}

\newcommand{\ai}[1]{{#1}}

\title{Pentagram maps over rings, Grassmannians, and skewers}
\author{Leaha Hand\thanks{
Department of Mathematics,
University of Arizona;
e-mail: {\tt leahahand@math.arizona.edu }
} \,
and Anton Izosimov\thanks{
School of Mathematics \& Statistics, University of Glasgow; e-mail:
{\tt anton.izosimov@glasgow.ac.uk}
} }
\date{}

\numberwithin{equation}{section}

\begin{document}
\maketitle
\abstract{
The pentagram map is a discrete dynamical system on planar polygons. By definition, the image of a polygon $P$ under the pentagram map is the polygon $P'$ whose vertices are the
intersection points of consecutive shortest diagonals of $P$. The pentagram map was introduced by R. Schwartz in 1992,
and is now one of the most renowned discrete integrable
systems. 

Several authors proposed generalizations of the pentagram map to other geometries, in particular to Grassmannians, where the role of points and lines is played by higher-dimensional subspaces, as well to {skewer geometry}, where both points and lines are affine lines in the three-dimensional Euclidean space.

In the present paper, we develop a common framework for these kinds of generalizations. Specifically, we show that those maps can be viewed as pentagram maps in the projective plane over an appropriate ring.  In general, those rings need not be division rings or commutative. We show that the Grassmannian pentagram map corresponds to the ring of matrices, while the skewer map is the pentagram map over the ring of dual numbers. Furthermore, we prove that the pentagram map remains integrable for {any} {stably finite} ground ring $R$.
}

\tableofcontents


\section{Introduction}
\subsection{Overview}

The pentagram map is a discrete dynamical system on the space of planar polygons introduced in \cite{pentagrammap}. 
The definition of the pentagram map is illustrated in Figure~\ref{fig:PM}: the image of the polygon $P$ under the pentagram map is the polygon $P'$ whose vertices are the intersection points of consecutive shortest diagonals of~$P$, i.e.  diagonals connecting second-nearest vertices.
In~\cite{integrabilityPMOvScTa, solovievintegrability}, it was shown that the pentagram map is a completely integrable system.  As of now, it is one of the most famous discrete integrable systems which features numerous connections to various areas of mathematics.  In particular, the pentagram map has an interpretation in terms of cluster algebras~\cite{glick2011pentagram}, networks of surfaces \cite{clusteralgebras}, the dimer model~\cite{affolter2024vector}, Poisson-Lie groups \cite{loopgroups}, and difference operators \cite{izopentagramandrefactorization}. Furthermore, the pentagram map can be viewed as a space-time discretization of the Boussinesq equation \cite{integrabilityPMOvScTa}, a shallow water approximation.

\begin{figure}[b]
\centering
\begin{tikzpicture}[]
\coordinate (VK7) at (0,0);
\coordinate (VK6) at (1.5,-0.5);
\coordinate (VK5) at (3,1);
\coordinate (VK4) at (3,2);
\coordinate (VK3) at (1,3);
\coordinate (VK2) at (-0.5,2.5);
\coordinate (VK1) at (-1,1.5);

\draw  [line width=0.5mm]  (VK7) -- (VK6) -- (VK5) -- (VK4) -- (VK3) -- (VK2) -- (VK1) -- cycle;
\draw [dashed, line width=0.2mm, name path=AC] (VK7) -- (VK5);
\draw [dashed,line width=0.2mm, name path=BD] (VK6) -- (VK4);
\draw [dashed,line width=0.2mm, name path=CE] (VK5) -- (VK3);
\draw [dashed,line width=0.2mm, name path=DF] (VK4) -- (VK2);
\draw [dashed,line width=0.2mm, name path=EG] (VK3) -- (VK1);
\draw [dashed,line width=0.2mm, name path=FA] (VK2) -- (VK7);
\draw [dashed,line width=0.2mm, name path=GB] (VK1) -- (VK6);

\path [name intersections={of=AC and BD,by=Bp}];
\path [name intersections={of=BD and CE,by=Cp}];
\path [name intersections={of=CE and DF,by=Dp}];
\path [name intersections={of=DF and EG,by=Ep}];
\path [name intersections={of=EG and FA,by=Fp}];
\path [name intersections={of=FA and GB,by=Gp}];
\path [name intersections={of=GB and AC,by=Ap}];

\draw  [line width=0.5mm]  (Ap) -- (Bp) -- (Cp) -- (Dp) -- (Ep) -- (Fp) -- (Gp) -- cycle;

\node at (-0.9,2.3) () {$P$};
\node at (2,1.5) () {$P'$};

\end{tikzpicture}
\caption{The pentagram map.}\label{fig:PM}
\end{figure}
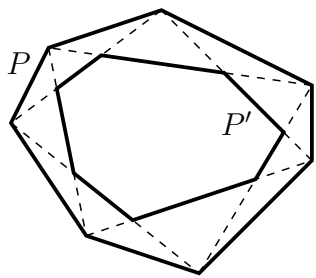

The pentagram map has been generalized in several different directions. One can intersect diagonals joining vertices that are further apart, consider polygons in higher-dimensional spaces, etc. \cite{ clusteralgebras, higherdimPM, higherdimPM2}.  Here, we are interested in \ai{other} kinds of generalizations  which are based on replacing conventional notions of points and lines by more general objects with similar properties. One example of such a map is the \emph{Grassmann pentagram map}, defined in~\cite{felipe2015pentagram} and further studied in \cite{noncommCylinder, ovenhousegrassman}. In this setting, the role of \emph{points} is played by elements of $\Gr(k,3k)$, i.e. $k$-dimensional subspaces in a $3k$-dimensional vector space, while \emph{lines} are elements of $\Gr(2k,3k)$. Any two {generic} points $V_1, V_2 \in \Gr(k,3k)$ belong to a single line $V_1 \oplus V_2 \in \Gr(2k,3k)$, and any two generic lines $W_1, W_2 \in \Gr(2k,3k)$ meet at a point $W_1 \cap W_2 \in \Gr(k,3k)$. The Grassmann pentagram map takes a (cyclically labeled) sequence of points $V_i \in  \Gr(k,3k)$, which we think of as a \emph{polygon \ai{in} the Grassmannian} $\Gr(k,3k)$, to a new sequence given by $(V_i \oplus V_{i+2}) \cap (V_{i+1} \oplus V_{i+3})$. In other words, just like the classical pentagram map, the Grassmann map is defined by intersecting diagonals connecting second-nearest vertices. Furthermore, the classical pentagram map is nothing else but the $k=1$ version of the Grassmann map.

Another generalization of the pentagram map, similar in spirit, is the \emph{skewer pentagram map} \cite{tabwineskins}. In this case, both points and lines are affine lines in the Euclidean space $\mathbb R^3$, with a point being incident to a line if and only if the corresponding affine lines meet a right angle. Thus, both the operation of constructing a line through two points and finding the intersection of two lines reduce to taking the common perpendicular, or \emph{skewer}. Therefore, the corresponding pentagram map  takes a cyclically labeled sequence $\ell_i$ of affine lines in $\mathbb R^3$ to the sequence $S(S(\ell_i , \ell_{i+2}) , S(\ell_{i+1} , \ell_{i+3}))$, where $S(\ell,\ell')$  stands for the skewer of lines $\ell, \ell'$.

In the present paper, we develop a common framework for these types of pentagram-like maps. Specifically, we show that such maps can be viewed as pentagram maps in the projective plane over an appropriate {ring}.  In general, those rings need not be division rings or commutative. We show that the Grassmann pentagram map corresponds to the ring of matrices, while the skewer map is the pentagram map over the ring $\R[\eps] / (\eps^2)$ of \emph{dual numbers}. 

Furthermore, we prove that the pentagram map is integrable for \emph{any} \emph{stably finite} ground ring $R$, thus generalizing the results of  \cite{felipe2015pentagram, ovenhousegrassman} on integrability of the Grassmann map and answering the question of \cite{tabwineskins} regarding integrability of the skewer map. The formal definition of stably finite rings will be given in Section \ref{sec:sfr}, but we emphasize here that this is a large class of rings which includes, amongst others, all commutative and all Noetherian rings. 

Integrability is understood in this paper as the existence of a \emph{Lax form}. A discrete dynamical system is said to admit a Lax form if  it can be written in the form $L \mapsto A^{-1}LA$ for certain operators $L$, $A$ (for details, see Section \ref{sec:lax}). In the setting of pentagram maps over a ring $R$, these operators are \emph{pseudo-difference operators} with coefficients in $R$. \ai{When $R$ is a field, we recover the Lax form of the classical pentagram map given in \cite{clusteralgebras, izopentagramandrefactorization}.}

The existence of a Lax form is one of the characteristic features of integrable dynamical systems. In particular, for the classical pentagram map (i.e., for the pentagram map over a field), the Lax form has been a central tool for establishing the map's dynamical properties, such as quasi-periodicity, cf. \cite{solovievintegrability, weinreich2023algebraic}. The study of dynamical aspects of integrability for pentagram maps over rings is beyond the scope of the present paper. However, we do expect that quasi-periodicity-type results can be extended from fields to  finite-dimensional algebras over fields. 

Our main results can be summarized as follows:
\begin{theorem}  
Let $R$ be a stably finite ring. Then there is a natural one-to-one correspondence between the set of projective equivalence classes of polygons in the left  projective space $\Proj^{d}(R)$ and an appropriate quotient of the space of degree $d+1$ left difference operators with coefficients in $R$. In terms of difference operators, the pentagram map in $\Proj^2(R)$ is described by a linear equation
$$
  \wD_+\D_-=\wD_-\D_+
  $$
  where the operators $\D_\pm$ are associated to a polygon, and $\wD_\pm$ to its image under the pentagram map. As a corollary, the pentagram map over $R$ admits a Lax representation.
\end{theorem}

We note that while similar results in the case of a field were obtained by the second author in \cite{izopentagramandrefactorization}, the arguments from that paper do not seem to generalize to the ring setting, and the construction of the present paper is essentially different. The main technical challenges to overcome are possible non-commutativity of the base ring along with the impossibility of resorting to general-position-type arguments.

We also note that while there are versions of the pentagram map in projective spaces of higher dimension, we chose to not pursue this direction here, as it comes with additional technical difficulties  related to more complicated combinatorics of multi-dimensional pentagram maps (intersection of a collection of planes as opposed to two lines).

\subsection{Structure of the paper}

In Section \ref{sec:sfr}, we recall the definition and standard properties of stably finite rings. Details can be found, e.g., in \cite[Chapter 1, \S 1B]{LamMR}. We also include a concise glossary, Table \ref{tab:glossary}, of terms from ring and module theory that are relevant to this paper. 

Section \ref{ProjSpaces} serves as an introduction into projective geometry over rings. We note that while ring geometry is a well-studied subject (see \cite{veldkamp1995geometry}), the setting of stably finite rings appears to be new. In Section~\ref{PS}, we define projective spaces as projectivizations of arbitrary modules. The case of free modules is treated in detail in Section \ref{sec:ndps}. We call the projectivization of a free rank $d+1$ module the \emph{projective space of dimension $d$}. \ai{We note that these projective spaces are \emph{not} projective spaces in the sense of incidence geometry (also known as \emph{combinatorial projective spaces}). In particular, for a general (stably finite) ring $R$, there are pairs of points in $\Proj^2(R)$ which are not contained in a unique line, as well as lines that do not share a unique point.}

Sections \ref{sec:psfmr} and \ref{sec:psfmr2} are devoted to projective spaces over the full matrix ring, which is the first of our two main examples. The second main example, the ring of dual numbers, is discussed in Section~\ref{sec:ppdn}. Additionally, in Section \ref{sec:path}, we discuss projective spaces over path algebras of quivers, including the algebra of triangular matrices.

In Section \ref{sec:ppm}, we define the notion of a \emph{polygon} in the projective plane over a ring, as well as show that the pentagram map on polygons is well-defined. Furthermore, we show that, for matrix rings and dual numbers, the general notion of a pentagram map over a ring specializes to Grassmann and skewer pentagram maps respectively. 

In Section \ref{DO}, we introduce our main technical tool: difference operators with coefficients in a ring. After discussing basic properties of such operators, we use that language to give an algebraic description of the moduli space of polygons in the projective \ai{space} over a ring. The main result of the section, Theorem \ref{DOcorresPOLY}, is a ring counterpart of \cite[Proposition 3.3]{izopentagramandrefactorization} which gives a description of real polygons in terms of difference operators. We note that the proof in the ring case is substantially more technical due to the lack of commutativity and fundamentally relies on the stable finiteness condition. 

In Section \ref{knownPM}, we establish integrability of the pentagram over an arbitrary stably finite ring. As a first step, in Section \ref{sec:pmdo}, we reformulate the pentagram map in the language of difference operators. The main result of the section is Theorem \ref{IPMintermsDO}, which generalizes  \cite[Theorem 1.1]{izopentagramandrefactorization} and says that, in terms of difference operators, the pentagram map is described by a \emph{linear} equation. This is used in Section \ref{sec:lax} to give a Lax form of the pentagram map. It is given in terms of pseudo-difference operators, which are introduced in Section \ref{PDO}.

Finally, in Section \ref{InvariantsGPM}, we use the Lax form to give invariants (i.e., conserved quantities, or \emph{first integrals}) of the ring pentagram map on closed polygons. These invariants are valued in the \emph{cyclic space} of $R$ which is a non-trivial Abelian group unless $R$ is a \emph{commutator ring}, i.e., is spanned over $\Z$ by expressions of the form $ab-ba$. In particular, we get non-trivial invariants for any commutative ring (whose cyclic space is the ring itself) and the full matrix ring (whose cyclic space coincides with the ground field).  In the case when the ground ring is $\R$, these invariants (in a slightly more general setting of \emph{twisted polygons}) are known to form a maximal Poisson-commuting  family~\cite{integrabilityPMOvScTa}. We expect that similar results can be obtained in the ring setting. 

\medskip
{\bf Acknowledgments.} \ai{The authors are grateful to Boris Khesin, Nicholas Ovenhouse, Richard Schwartz,  Max Weinreich, and anonymous referees for useful comments and fruitful discussions.} This work was supported by NSF grant DMS-2008021 and the Simons Foundation through its Travel Support for Mathematicians program. A part of this work was done during A.I.’s visit to Max Planck Institute for Mathematics in Bonn. A.I. would like to thank the Institute’s faculty and staff for their support and stimulating atmosphere.

\section{Stably finite rings}\label{sec:sfr}
\begin{table}
\renewcommand{\arraystretch}{1.25}
\centering
\footnotesize
\begin{tabular}{l l}
    \hline
    \textbf{Unital ring} & Ring with multiplicative identity. \\ 
    \textbf{Opposite ring} &  
    {Ring with same elements and addition but reverse multiplication order.} \\
    \textbf{Dedekind-finite ring} & Ring for which $ab = 1 \implies ba = 1$. \\
    \textbf{Left Noetherian ring} & Ring for which any ascending chain of left ideals must stabilize.\\
        \textbf{Cyclic space of a ring} & Quotient $R/[R,R]$, where $[R,R]$ is the $\Z$-span of commutators.\\
               \textbf{Commutator ring} & Ring $R$ such that $[R,R] = R$.\\
    \textbf{Endomorphism} & Homomorphism from a module to itself. \\
    \textbf{Split exact sequence} & Short exact sequence isomorphic to $0\to A \to A\oplus B\to B \to 0$.  
    \\
    \textbf{Unital module} & Module $M$ over unital ring such that $1v = v$ for all $v \in M$. \\
        \textbf{Cyclic module} & Module generated by a single element.\\
    \textbf{Free module} & Module that has a basis (need not be finite in length).\\
    \textbf{Finite-rank free module} & Free module that admits a basis of finite length. \\
    \textbf{Hopfian module} & Module for which any surjective endomorphism is also  an isomorphism. \\
    \textbf{Noetherian module} & Module for which any ascending chain of submodules must stabilize. \\
    \textbf{Stably finite ring} & Ring for which any finite-rank free module is Hopfian.  \\
        \textbf{Ring with IBN property} & Ring for which all bases of a finite-rank free module have same length. \\ 
     \textbf{Rank of a free module} & Number of elements in a basis. Well-defined for rings with IBN property.  \\ 
    \hline
\end{tabular}
\caption{Glossary of terms related to rings and modules}
\label{tab:glossary}
\end{table}

In this section, we recall the definition and basic properties of \emph{stably finite rings} -- the main class of rings relevant to the present paper; for details, see \cite[Chapter 1, \S 1B]{LamMR}. One of the main results of the section is that such rings have an \emph{invariant basis number} property, which means that any finitely generated free module over such a ring has a well-defined \emph{rank}, just like any vector space over a field has well-defined dimension.

For the reader's convenience, we include a concise table of relevant terms from ring theory that will be used frequently throughout the paper in Table \ref{tab:glossary}.  More detailed definitions and examples are given throughout the paper. In what follows, all rings are assumed to be unital and non-trivial (i.e. have a multiplicative identity element $1 \neq 0$), but may be non-commutative and/or have zero divisors. All modules are supposed to be unitary. 

\begin{definition}
	Let $R$ be a ring. An $R$-module $M$ is \emph{Hopfian} if any surjective endomorphism $\phi \in \End_R (M)$ is an isomorphism. 
\end{definition}

\begin{example}
A finite-dimensional vector space is a Hopfian module. The vector space $\R[x]$ of single-variable polynomials is not Hopfian, since the differentiation map is surjective but not an isomorphism.
\end{example}
For free modules, the Hopfian condition can also be reformulated using the notion of a \emph{Dedekind\ai{-}finite} ring. Recall that a ring is {Dedekind\ai{-}finite} if $ab = 1$ implies $ba = 1$.  

\begin{example}
A common example of a non-Dedekind-finite ring is the ring of endomorphisms of the vector space $\R^{\N}$ of real-valued sequences $x_1, x_2, \dots \in \R$. Consider left shift and right shift operators on $\R^{\N}$, defined by
$
L(x_1, x_2, \dots) = x_2, x_3, \dots$, $ R(x_1, x_2, \dots) = 0, x_1, x_2, \dots
$.
Then $L \circ R$ is the identity, but $R \circ L$ is not.  
\end{example}
\begin{prop}\label{prop:hdf}
 Suppose $R$ is a ring, and $M$ is a free $R$-module. Then the following conditions are equivalent. 
 \begin{enumerate}
\item $M$ is Hopfian; \item the ring  $\End_R (M)$ is Dedekind\ai{-}finite 
\item if $M \simeq M \oplus N$ for some other $R$-module $N$, then $N = 0$.
 \end{enumerate}
\end{prop}

In what follows, the $\oplus$ symbol might mean external or internal direct sum, depending on the context. Here, $M \simeq M \oplus N$  means that $M$ is isomorphic to the external direct sum of $M$ and $N$, since an internal direct sum would not make sense in this setting. On the other hand, writing $M = N \oplus O$ for a given module $M$ would mean that $M$ is the internal direct sum of its submodules $N$ and $O$.

\begin{proof}[Proof of Proposition \ref{prop:hdf}]
	$1 \Rightarrow 2$. Suppose $M$ is Hopfian, and let $\phi, \psi \in \End_R (M)$ be such that $\phi\psi = \mathrm{id}$. Then $\phi$ is surjective and hence an isomorphism. So,  $\psi \phi= \mathrm{id}$ as well, as needed. 
	
	$2 \Rightarrow 3$.
	Suppose the ring $\End_R (M)$ is Dedekind\ai{-}finite, and let $\phi \colon M \to M \oplus N$ be an isomorphism. Let also $\pi \colon M \oplus N \to M$ be the projection onto the first summand, and $i \colon M \to M \oplus N$ be the inclusion map. Then $( \pi \phi)( \phi^{-1} i) =  \pi i = \mathrm{id}$, so, by Dedekind\ai{-}finiteness, $\pi \phi$ is an isomorphism, and $N = \Ker \pi = 0$.
	
	$3 \Rightarrow 1$. Suppose that $M \simeq M \oplus N$ implies $N = 0$. Let $\phi \in \End_R (M)$ be surjective. Then, since $M$ is free, the exact sequence $0 \to \Ker \phi \to M \to M \to 0$ splits, and $M \simeq M \oplus \Ker \phi$. So, $\Ker \phi = 0$, proving that $M$ is Hopfian. \qedhere
	
\end{proof}
\begin{cor}\label{cor:dual}
Suppose  $R$ is a ring, and $M$ is a free $R$-module of finite rank \ai{(i.e., admits a finite basis)}. Then $M$ is Hopfian if and only if its dual module $M^* = \Hom_R(M, R)$ is Hopfian.
\end{cor}
\begin{proof}
For a free module of finite rank, the map sending an endomorphism to its adjoint is an anti-isomorphism between the rings $\End_R (M)$ and $\End_R (M^*)$. Since the Dedekind\ai{-}finiteness condition is left-right symmetric, it is preserved by anti-isomorphisms. 
\end{proof}

Furthermore, if $M$ is {free of finite rank}, the Hopfian condition can be restated as follows.
\begin{prop}\label{prop:gsb}
Suppose an $R$-module $M$ has a basis of size $n$. Then $M$ is Hopfian if and only if any  generating set of size $n$ is a basis of $M$. 
\end{prop}
\begin{proof}
Let $v_1, \dots, v_n$ be a basis of $M$. Suppose $M$ is Hopfian, and let $w_1, \dots, w_n$ be a generating set. Then the endomorphism $M \to M$  taking $v_i$ to $w_i$ is surjective and hence an isomorphism. Therefore, $w_1, \dots, w_n$ is a basis.\par
Conversely, assume that any generating set of size $n$ is a basis of $M$, and let $\phi \in \End_R (M)$ be surjective. Then $\phi$ takes $v_1, \dots, v_n$ to a generating set, which, by our assumption, must be a basis again. So, $\phi$ is an isomorphism.  
\end{proof}
 
\begin{definition} \label{defStablyFinite}
A ring $R$ is called \textit{stably finite} (or \emph{weakly finite}) if any free finite rank $R$-module is Hopfian. 

\end{definition}

Since the dual of a left module is a right module and vice versa, by Corollary \ref{cor:dual} we see that all  free finite rank left $R$-modules are Hopfian if and only if  all  free finite rank right $R$-modules are Hopfian. So, Definition \ref{defStablyFinite} might be phrased in terms of left modules or right modules. 

One of the most important properties of stably finite rings is that they have an \emph{invariant basis number} (IBN) property:

\begin{prop}\label{prop:ibn}
Assume that $R$ is stably finite, $M$ is an $R$-module, and $v_1, \dots, v_n$, $w_1, \dots, w_m$ are two bases in $M$. Then $ m = n$. 
\end{prop}
\begin{proof}
This is equivalent to saying that  $R^m \simeq R^n$ (as left or right modules) implies $m = n$. 
Assume that $ \phi \colon R^n \to R^m$ is an isomorphism. Without loss of generality, let $n \leq m$. Then the projection map $\pi \colon R^m \to R^n$, $\pi(r_1, \dots, r_m) := (r_1, \dots, r_n)$ is surjective, and so is the composition $\pi\phi \colon R^n \to R^n$. Therefore, $\pi \phi$ is an isomorphism, and $\pi$ is injective. So, $ m = n$. 
\end{proof}
Thus, for every finite rank free module over a stably finite ring, the notion of rank is well-defined.

The following gives another equivalent characterization of stably finite rings. 
\begin{prop} \label{epiisiso}
 A ring $R$ is  stably finite if and only if all matrix rings $ \text{M}_n(R)$ are Dedekind-finite. In particular, all stably finite rings are Dedekind-finite. 

\end{prop}
\begin{proof}
The first claim follows from Proposition \ref{prop:hdf}, combined with the isomorphism $\End_R (M) \simeq \Mat_n(R)$ for a free rank $n$ right module $M$. The second claim follows from the identification $ \Mat_1(R) \simeq R$. \qedhere
\end{proof}

 Recall that a ring is \emph{left Noetherian} if any ascending chain $I_1 \subset I_2 \subset \dots$ of left ideals eventually stabilizes, i.e., $I_k = I_{k+1} = \dots$ for some $k$. Most rings encountered in applications are Noetherian.

 \begin{example}
     Any finitely generated algebra over a field is a Noetherian ring. On the other hand, the ring $ \R[x_1, \dots]$ of polynomials in infinitely many variables is not Noetherian, since we have a strictly ascending chain of ideals $\langle x_1 \rangle \subset \langle x_1, x_2 \rangle \subset \dots$
 \end{example}

\begin{prop} \label{thmKindsofStablyFiniteRings}
All left Noetherian rings are stably finite. 
\end{prop}
\begin{proof} 
Let $M$ be a finite rank free left module over a left Noetherian ring $R$. Then $M$ is a Noetherian module. Let $\phi \in \End_R(M)$ be surjective. By the Noetherian property, the chain of submodules $\Ker \phi \subset \Ker \phi^2 \subset \dots$ must stabilize: $\Ker \phi^n = \Ker \phi^{n+1}$ for some $n$. Suppose $v \in \Ker \phi$. By surjectivity, there is $w \in M$ such that $v = \phi^n(w)$. Then $\phi^{n+1}(w) = \phi(v) = 0$, so $w \in \Ker \phi^{n+1}$. But $\Ker \phi^n = \Ker \phi^{n+1}$, so $v = \phi^n(w) = 0$. Thus, $\phi$ is an isomorphism. 

\end{proof}
Likewise, all right Noetherian rings are stably finite.

\section{Projective spaces over rings} \label{ProjSpaces}

This section is an introduction into projective geometry over rings. We note that while ring geometry is a well-studied subject, 
the setting of stably finite rings considered here appears to be more general than what can be found in the existing literature. For instance, \cite{veldkampSR2, veldkamp1995geometry} only consider rings of \emph{stable rank $2$}. Any such ring is stably finite, but not every stable finite ring has stable rank $2$. For instance, the ring of integers, which is commutative and hence stably finite, has stable rank $3$. Furthermore, there exist commutative, and hence stably finite, rings of arbitrary finite and even infinite stable rank \cite{mortini1992example, vass}. 


\subsection{The projectivization of a module} \label{PS}

There exist several, generally inequivalent, ways to define projective spaces over rings. Here we adapt the definition given in \cite[Section 2]{veldkampSR2}. Although we will mainly be interested in projective spaces obtained from free modules, we first define the projectivization of an arbitrary module. 
We only consider the case of left modules. Right modules are considered analogously.

\begin{definition} \label{defSubspace}

Let $M$ be an $R$-module. A \emph{subspace}  $N \subset M$ is a free submodule which admits a direct complement in $M$ (i.e.\ai{,} there is a submodule $K \subset M$, not necessarily free, such that $M = N \oplus K$).
\end{definition}

\begin{example}
Consider $\Z$ as a module over itself. $2\Z \subset \Z$ is its free submodule but not a subspace, since it does not admit a complement.
\end{example}

\begin{prop}\label{prop:subspsm}
Suppose $M$ is an $R$-module, and $N \subset M$ is a subspace. Then $N$ is also a subspace in any submodule of $M$ containing $N$.
\end{prop}
\begin{proof}
If $M = N \oplus K$, and $L \supset N$ is another submodule of $M$, then $L = N \oplus (K \cap L)$.
\end{proof}

\begin{definition}
 The \emph{projectivization} $\Proj (M)$ is the set of rank one subspaces of $M$. 
\end{definition}
\begin{definition}
Given $R$-modules $M, N$, a \emph{projective transformation} $\Proj (M) \to \Proj (N)$  is a map induced by a module isomorphism $M \to N$.
\end{definition}

For a module $M$, its projectivization $\Proj(M)$ may also be defined as the set of \emph{unimodular} elements modulo scalars. 

\begin{definition}\label{def:unimod}
Let $M$ be an $R$-module, and let $M^* =\Hom_R(M,R)$ be its dual module. An element $v \in M$ is called \emph{unimodular} if there exists $\xi \in M^*$ such that $\xi(v) = 1$ (equivalently, if there exists $\xi \in M^*$ such that $\xi(v)$ is a unit).
\end{definition}
\begin{example}
In the $\Z$-module $\Z^n$, an element $(k_1, \dots, k_n)$ is unimodular if and only if $\mathrm{gcd}(k_1, \dots, k_n) = 1$.
\end{example}
\begin{prop} \label{1DandUni}
Rank one subspaces of  \ai{an $R$-module} $M$ are precisely cyclic submodules of $M$ generated by unimodular elements.
\end{prop}
\begin{proof}
Suppose $P \subset M$ is a rank one subspace. Then, by definition, $P$ is freely generated over $R$ by some $v \in M$, and is a direct summand in $M$, i.e. $M = P\oplus N$. Therefore, for each $w \in M$ there exists unique $r \in R$ and $u\in N$ such that $w=  r v+ u$. Define a homomorphism $\xi \colon M \to R$ by $\xi(r v+ u) :=  r$.  Then $\xi(v) = 1$, ensuring $v$ is unimodular. 

Now, suppose that $P = Rv$ for some unimodular $v \in M$. Since $v$ is unimodular, there exists some $\xi \in M^*$ such that $\xi(v) = 1$. Then $\xi(rv) = r$, so $rv \neq 0$ as long as $r \neq 0$, meaning that $P$ is free. Furthermore, we have $M = P \oplus \Ker \xi$, so $P$ is indeed a rank one subspace. \qedhere

\end{proof}

\begin{prop}
Two unimodular elements $v, w $ \ai{of an $R$-module $M$} generate the same subspace if and only if they are unit multiples of each other.
\end{prop}
\begin{proof}
 Clearly, if $r$ is a unit, then $Rv = Rrv$. Conversely, assume that $Rv = Rw$. Then $w \in Rv$, so $w = rv$ for some $r \in R$. Likewise, $v = r'w$ for some $r' \in R$. So, $w = rv = rr'w$. Since $w$ is unimodular, there is $\xi \in M^*$ such that $\xi(w) = 1$. So,
$$
rr' = rr'\xi(w) \stackrel{(*)}{=} \xi(rr'w) = \xi(w) = 1,
$$
where $(*)$ holds because $\xi \colon M \to R$ is a homomorphism of left $R$-modules. 
Similarly, $r'r = 1$. So, $r$ is a unit. 
\end{proof}

\begin{cor}\label{cor:qbu}
\ai{For any $R$-module $M$,} there is a one-to-one correspondence between points of $\Proj(M)$ and unimodular elements in $M$ modulo units. 
\end{cor}

\subsection{Projective spaces over the full matrix ring}\label{sec:psfmr}

Here we consider our first main example -- the projective space over the full matrix ring, cf. \cite[Section 7]{veldkamp1995geometry} and \cite{thas1971m}. 

Let $\F$ be a field, and $R = \Mat_k(\F)$ be the ring of $k \times k$ matrices over $\F$. By definition, projective spaces over $R$ are projectivizations of $R$-modules. Since $R$ is isomorphic to its opposite ring via transposition, there is no distinction between left and right modules. So, we will restrict our attention to left modules. Any left $\Mat_k(\F)$-module $M$ is isomorphic to the one of the form $V^k = \F^k \otimes_\F V$, where $V$ is a (not necessarily finite-dimensional) vector space over $\F$, and $\Mat_k(\F)$ acts on the first tensor factor\footnote{Note that, as an $\F$-vector space, $V^k$ is just a direct sum of $k$ copies of $V$, hence the notation. We define it as  $\F^k \otimes_\F V$ to emphasize the $\Mat_k(\F)$-module structure. The correspondence $V \mapsto V^k$ provides \ai{an} additive equivalence between the categories of $\F$-modules (vector spaces) and $\Mat_k(\F)$-modules. This is known as the  \emph{Morita equivalence} between the corresponding rings $\F$ and $\Mat_k(\F)$.}. Thus, for the purpose of describing projective spaces over $\Mat_k(\F)$,  it suffices to consider left modules of the form  $M = V^k$, where $V$ is a vector space over $\F$.

\begin{prop}\label{prop:prgr}
There is a natural identification $\Proj(V^k) \simeq \Gr(k, V)$.
\end{prop}
\begin{proof}
Submodules of the module $V^k$ are of the form
$
W^k = \F^k \otimes_\F W
$
where $W \subset V$ is an arbitrary subspace of $V$. The correspondence $W \mapsto W^k$ gives an isomorphism between the lattice of subspaces of $V$ and the lattice of submodules of $V^k$. In particular, any submodule of $V^k$ admits a direct complement, which means that subspaces of $V^k$ are precisely its free submodules. Further, $W^k$ is a free rank one $\Mat_k(\F)$-module precisely when $\dim_\F W = k$, so rank one free submodules of $V^k$ are \ai{those} of the form $W^k$ for $W \in \Gr(k, V)$.
\end{proof}

\begin{remark}
By Corollary \ref{cor:qbu}, the space $\Proj(M)$ can alternatively be described as the set of unimodular elements in $M$ modulo units. For a vector space $V$ of finite dimension $n$, the module  $V^k$ is isomorphic to the module $\Mat_{k,n}(\F)$ of $k \times n$ matrices. Such a matrix is unimodular in the sense of Definition \ref{def:unimod} if and only if it has full rank. Thus $$\Proj(V^k) = \{ \mbox{full rank}\, k \times n \mbox{ matrices}\}/ \{ \mbox{invertible}\, k \times k \mbox{ matrices}\},$$ which is the standard description of $\Gr(k, n)$.
\end{remark}

\subsection{The \textit{d}-dimensional projective space}\label{sec:ndps}

\begin{definition}
Suppose that $M$ is left free $R$-module of rank $d+1$. Then we call $\Proj (M)$ a (left) \emph{projective space of dimension $d$}.
\end{definition}

Clearly, all (left) projective spaces of the same (finite) dimension are projectively equivalent (i.e., are related to each other by a projective transformation). We will frequently not distinguish between projectively equivalent objects and talk about the \emph{$d$-dimensional projective space over $R$}, denoted $\Proj^d(R)$. Without loss of generality, one can think of $\Proj^d(R)$ as the projectivization of $R^{d+1}$.

In addition to points in projective spaces $\Proj(M)$, we will be interested in higher-dimensional subspaces, corresponding to subspaces of $M$ of rank $k \geq 2$. For stably finite rings, one has the following:

\begin{prop}\label{prop:rankBound}
   If $M$ is a free finite rank module over a stably finite ring $R$, and $N \subsetneq M$ is a proper subspace, then $\rank(N) < \rank(M)$.
   
\end{prop}
\begin{proof}
Assume $N \subsetneq M$ is a subpace. If $\rank(N) \geq \rank(M)$ (including if $N$ has an infinite rank), then $N$ has a submodule isomorphic to $M$ which is a direct summand in $N$ and hence in $M$. But this is impossible since $M$ is Hopfian (see Proposition \ref{prop:hdf}). 
\end{proof}

\begin{cor}\label{cor:di}
Let $R$ be a stably finite ring, and $M$ be an $R$-module. If $L \subsetneq N$ are subspaces of $M$, then $\rank(L) < \rank(N)$.
\end{cor}
\begin{proof}
By Proposition \ref{prop:subspsm}, we have that $L$ is a subspace of $N$, so the result follows from Proposition~\ref{prop:rankBound}. 
\end{proof}

In particular, subspaces of the projective plane $\Proj^2(R)$ over a stably finite ring are its points (i.e., rank one subspaces in a free rank three $R$-module $M$) and \emph{lines} (i.e., rank two subspaces $L \subset M$). No point is contained in another point, and no line is contained in another line. So, the poset of subspaces of $\Proj^2(R)$ is determined by the incidence of points and lines. Such a structure is known as an \emph{incidence structure}.

\begin{definition} \label{defIncidenceStructure}
    An \textit{incidence structure} is a set of \emph{points}, denoted by $\Points$, a set of \emph{lines}, denoted $\Lines$, and an incidence $\Inc \subset \Points \times \Lines$. We say $P\in \Points, L \in \Lines$ are \emph{incident} and write $P \in L$ if  $(P,L) \in \Inc$.  Two incidence structures $(\Points, \Lines, \Inc)$ and $(\Points', \Lines', \Inc')$ are isomorphic if there are bijections $\Points \to \Points'$, $\Lines \to \Lines'$ such that the image of $\Inc$ under the induced map $\Points \times \Lines \to \Points' \times \Lines' $ is precisely $\Inc'$. 

\end{definition}
If $\Points$ is a point set of some (fixed) incidence structure, we will say that $\Points$ is an \emph{incidence space}. We will say that a map $\Points \to \Points'$ is an isomorphism of incidence spaces if it extends to an isomorphism of incidence structures.\par 

We may now formalize how the projective plane $\Proj^2(R) = \Proj(M)$, where $M$ is a rank three free $R$-module, is an incidence space: the set of lines is
\begin{gather*} 
 \Lines( \Proj(M)) := \{ \mbox{rank two subspaces of $M$}\}, 
\end{gather*}
and the incidence $\Inc \subset  \Proj(M) \times \Lines( \Proj(M))$ is given by containment ($
(P, L) \in \Inc \iff P \subset L$).

Clearly, any projective transformation $\Proj(M) \to \Proj(N)$, where $M, N$ are rank three free $R$-modules, is also an isomorphism of incidence spaces. Thus, all (left) projective planes over $R$ are isomorphic as incidence spaces, and one can talk about \emph{the natural incidence structure on $\Proj^2(R)$}. In a similar way, one can define \emph{higher rank incidence structures} on projective spaces $\Proj^d(R)$, $d > 2$.

\subsection{The projective plane over the full matrix ring}\label{sec:psfmr2}
Back to the example $R = \Mat_k(\F)$, notice that $V^k$ is a free  $\Mat_k(\F)$-module of rank $d+1$ if and only if $\dim V = k(d+1)$. Thus, by Proposition \ref{prop:prgr}, we have 
$$
\Proj^d(\Mat_k(\F)) = \Gr(k, (d+1)k).
$$
In particular, 
$$
\Proj^2(\Mat_k(\F)) = \Gr(k, 3k).
$$
Furthermore, the space of lines in $\Proj^2(\Mat_k(\F))$ can be described as follows. 
\begin{prop}\label{prop:matrixplane}
There is a natural identification $\Lines(\Proj^2(\Mat_k(\F))) \simeq \Gr(2k, 3k)$. The incidence relation between points $P \in \Gr(k, 3k)$ and lines $L \in \Gr(2k, 3k)$ is given by containment. 
\end{prop}
\begin{proof}
We have $\Proj^2(\Mat_k(\F))  = \Proj(V^k)$, where $V$ is a vector space over $\F$ of dimension $3k$. Recall that the correspondence $W \mapsto W^k$ identifies the lattice of subspaces of $V$ with lattice of submodules of $V^k$. Furthermore, $W^k$ is free of rank two if and only if $\dim W = 2k$. Thus, lines in $ \Proj(V^k)$ correspond to $2k$-dimensional subspaces of $V$, and incidence between points and lines is given by containment. 
\end{proof}
In words, $\Proj^2(\Mat_k(\F))$ can be described as follows: points are $k$-dimensional subspaces of~$\F^{3k}$, lines are $2k$-dimensional subspaces of $\F^{3k}$, while incidence of points and lines corresponds to containment. Note that two \emph{generic} points $P_1, P_2 \in \Proj^2(\Mat_k(\F))$ lie on a unique line $P_1 \oplus P_2$. Likewise, two \emph{generic} lines $L_1, L_2$ meet at a unique point $L_1 \cap L_2$. However, if $\dim P_1 \oplus P_2 < 2k $, then there exist infinitely many lines containing points $P_1, P_2$ (and, if $\dim L_1 \cap L_2 > k$, there exist infinitely many points lying on $L_1, L_2$). Thus, $\Proj^2(\Mat_k(\F))$ is  not a \emph{\ai{combinatorial} projective plane}. Incidence spaces with properties similar to those of  $\Proj^2(\Mat_k(\F))$ are known as \emph{Barbilian planes}, see \cite{barbilian1940axiomatik, veldkampSR2}. It is shown  in \cite[Section 5.12]{veldkampSR2} that such planes are precisely the projective planes over rings of \emph{stable rank two}, of which the matrix ring is an example. 

\subsection{The projective plane over the dual numbers}\label{sec:ppdn}

Let $\DN := \R[\eps] / (\eps^2)$ be the ring of dual numbers. In this section, we show that the geometry of $\Proj^2(\DN)$ is the \emph{skewer geometry} of \cite{tabskewers}, where both points and lines are affine lines in $\R^3$, and a point is incident to a line if and only if the corresponding affine lines meet at a right angle. We note that the relationship between dual numbers and lines in $\R^3$ has been known since the introduction of the former; see, e.g., \cite[Section 2.3]{pottmann2001computational} and \cite[Section 2.6]{tabskewers}. However, quite strikingly, the result on the precise coincidence of the corresponding geometries does not seem to appear in the existing literature. 

Let $V$ be a three-dimensional vector space over $\R$. 
Then $\DN \otimes_\R V$ is a free rank three $\DN$-module. Our aim is to give a geometric description of the incidence space $\Proj(\DN \otimes_\R V)$, to which end we fix a positive-definite inner product on $V$. Of course, the incidence structure on $\Proj(\DN \otimes_\R V)$ does not depend on that inner product. However, our geometric model does. 

Let $\AL(V)$ be the space of affine straight lines in $V$.   
We define an incidence structure on $\AL(V)$ as follows: both points and lines are \ai{affine lines in $V$, i.e., $\Points= \Lines = \AL(V)$,} and a point is incident to a line if and only if the affine lines intersect and are perpendicular to each other. In this geometry, both the operation of joining two points with a line and intersecting two lines boil down to taking the common perpendicular, or \emph{skewer}, of the corresponding affine lines.

We will show that there is natural isomorphism of incidence spaces $\Proj(\DN \otimes_\R V) \simeq  \AL(V)$. The corresponding mapping $\phi_{P} \colon \AL(V) \to \Proj(\DN \otimes_\R V)$ between the point spaces is a variant of the \emph{Study mapping}, see \cite[Section 2.3]{pottmann2001computational} and \cite{Study}. For an affine line $\{ p + tv \mid t \in \R\} $, where $p,v \in V$ and $v \neq 0$, let
\begin{equation}\label{eq:zvp}
\zeta(v,p) := v + \eps v \times p,
\end{equation}
where the notation  $\times$ stands for  the cross product in $V$, and define
$$
\phi_{\mathrm P}(p + tv) := \DN \zeta(v,p).
$$

\begin{prop}\label{prop:phip}
The mapping $\phi_{\mathrm P}$ is well-defined and bijective. 
\end{prop}
\begin{proof}
First, note that a dual vector $v + \eps w \in  \DN \otimes_\R V$ (where $v,w \in V$) is unimodular if and only if $v \neq 0$. Indeed, if $v \neq 0$, then, for the $\DN$-linear extension of any $\xi \in V^*$ such that $\xi(v) \neq 0$, we have that $\xi(v + \eps w)$  is a unit. Conversely, if $v = 0$, then for any $\xi \in  (\DN \otimes_\R V)^*$, we have that $\xi(v +  \eps w)$ is divisible by $\eps$ and thus is not a unit. 

In particular, for any $v \neq 0$  we have that $\zeta(v,p)$ is unimodular and therefore $ \DN \zeta(v,p)$ is a subspace. 
Moreover, this subspace \ai{does not depend on} the choice of the direction vector $v$ and base point $p$ of the line $v + tp$, since for any $\lambda \in \R$ we have
$\zeta(\lambda v, p) = \lambda \zeta(v, p)$ and $\zeta(v, p + \lambda v) = \zeta(v,p)$. 
So, the mapping $\phi_{\mathrm P}$ is well-defined. 

Verification of bijectivity is straightforward. \qedhere

\end{proof}

\begin{prop}\label{prop:freefree}
Let $M$ be a finite rank free $\DN$-module, and suppose that $M = N \oplus K$. Then $N$ and $K$ are free modules.
\end{prop}
\begin{proof}
A $\DN$-module is an $\R[\eps]$-module annihilated by $\eps^2$. Therefore, any finitely generated $\DN$-module $M$ is isomorphic to a unique module of the form $M^{p,q} = \DN^{\oplus p} \oplus (\DN / (\eps))^{\oplus q}$ (this is just the Jordan decomposition over the reals for the multiplication by $\eps$ in $M$). Such a module is free if and only if $q = 0$. So, the direct sum of two \ai{finitely generated $\DN$-}modules, at least one of which is not free, cannot be free. 
\end{proof}

\begin{prop}\label{prop:oc}
Suppose $M$ is a free $\DN$-module of rank $n$, equipped with a non-degenerate $\DN$-bilinear inner product.  Let $N \subset M$ be a subspace of rank $k$. Then $N^\bot := \{v \in  M \mid \langle v, w \rangle = 0 \mbox{ for all } w \in N\}$ is also a subspace\ai{, and $\mathrm{rank}\, N^\bot = n - k$}.
\end{prop}
\begin{proof}

Consider the restriction mapping $\rho \colon M^* \to N^* $. Since $N$ is free of rank $k$, so is $N^*$, and $M^* = \Ker \rho \oplus K$ for some submodule $K \subset M^*$. Thus, by Proposition \ref{prop:freefree}, the submodule $ \Ker \rho \subset M\ai{^*}$ is free. 
Furthermore, $\rho_K \colon K \to N^*$ is an isomorphism, so $K$ has rank $k$ and $\Ker \rho$ has rank $n - k$. Now, consider the mapping $\phi \colon M \to M^*$ given by the inner product. Since the inner product is non-degenerate, the mapping $\phi$ is injective. Therefore, since $M$ is a finite-dimensional real algebra, $\phi$ is, in fact, an isomorphism. Together, $N^\bot = \phi^{-1}(\Ker \rho)$ is a free submodule of $M$ which has rank $n - k$ and admits a direct complement, $\phi^{-1}(K)$, as needed. 
\end{proof}

\begin{cor}
Suppose $M$ is a free $\DN$-module of finite rank, equipped with a non-degenerate $\DN$-bilinear inner product.  Let $N \subset M$ be a subspace. Then $(N^\bot)^\bot = N$. 
\end{cor}
\begin{proof}
We have $(N^\bot)^\bot \supset N$. At the same time, by Proposition \ref{prop:oc}, the subspaces $(N^\bot)^\bot$ and $ N$ are of the same (finite) rank. Therefore, by Corollary \ref{cor:di}, we have $(N^\bot)^\bot =  N$. (Note that $\DN$ is a finite-dimensional real algebra, hence Noetherian, hence stably finite.) \qedhere

\end{proof}

\begin{cor}\label{cor:duality}
Suppose $M$ is a free $\DN$-module of rank $n$, equipped with a non-degenerate $\DN$-bilinear inner product. Then the correspondence $N \mapsto N^\bot$ is a bijection:
$$
\{ \text{subspaces of } M \mbox{of rank } k\} \leftrightarrow \{ \text{subspaces of } M \mbox{of rank } n - k\}.
$$
\end{cor}

We now define a map $\phi_{\mathrm L} \colon \AL(V) \to \Lines(\Proj(\DN \otimes_\R V))$ by
$$
\phi_{\mathrm L}(p + tv) := \phi_{\mathrm P}(p+tv)^\bot,
$$
where the inner product on $\DN \otimes_\R V$ is defined as the unique $\DN$-bilinear extension of the inner product on $V$. 
It follows from Proposition \ref{prop:phip} and Corollary \ref{cor:duality} that the map $\phi_{\mathrm L}$ is well-defined and bijective. 

\begin{prop}\label{prop:iis}
The mappings $\phi_{P} \colon \AL(V) \to \Proj(\DN \otimes_\R V)$ and $\phi_{\mathrm L} \colon \AL(V) \to \Lines(\Proj(\DN \otimes_\R V))$
provide an isomorphism of incidence spaces $\Proj(\DN \otimes_\R V) \simeq  \AL(V)$.
\end{prop}
\begin{proof}
We already know that the mappings $\phi_{\mathrm P}$, $\phi_{\mathrm L}$ are bijections, so it remains to show that they preserve incidence. In other words, we need to show that the lines $p + tv$ and $q + tw$ meet at a a right angle if and only if $\phi_{\mathrm P}(p + tv) \subset \phi_{\mathrm L}(q + tw)$. We have
$$
\phi_{\mathrm P}(p + tv) \subset \phi_{\mathrm L}(q + tw) = \phi_{\mathrm P}(q + tw)^\bot \quad \Leftrightarrow \quad \langle \zeta(p + tv), \zeta(q+tw) \rangle = 0.
$$
Furthermore,
$$
 \langle \zeta(p + tv), \zeta(q+tw) \rangle = \langle v + \eps v \times p, w + \eps w \times q \rangle = \langle v, w \rangle + \eps \det(v,w,q-p).
$$
But $\det(v,w,q-p) = 0$ if either $v$ and $w$ are collinear (which can not happen if $\langle v,w \rangle  = 0$, since the inner product on $V$ is positive-definite), or $q - p \in \spanof\langle v, w\rangle$. The latter is equivalent to saying that the lines $p + tv$ and $q + tw$ have a point in common. So, combining this with the condition $\langle v,w \rangle  = 0$, we get that  $\phi_{\mathrm P}(p + tv) \subset \phi_{\mathrm L}(q + tw)$ if and only if  the lines $p + tv$ and $q + tw$ meet at a right angle, as needed. 
\end{proof}

\begin{remark}
We note that, when restricted to lines through the origin, the map $\phi_{\mathrm P}$ is essentially the identity map on $\RP^2$, as it sends such a line to its direction; as for $\phi_{\mathrm L}$, it is the \emph{polar duality} map sending a line to its orthogonal plane. In this case, Proposition \ref{prop:iis} just says that points $P_1, P_2 \in \RP^2$ are polar if and only if $P_1$ belongs to the polar line of $P_2$. 
\end{remark}

\begin{remark}
We note that, while the ring of dual numbers can be defined over any field $\F$ by $\DN(\F) := \F[\eps]/(\eps^2)$, the construction of this section does not work for an arbitrary field. For example, for $\F = \C$, the  mapping $ \phi_{\mathrm P} \colon \AL(V) \to \Proj(\DN(\C) \otimes_\C V)$ which takes a line $p+tv$ in a three-dimensional complex vector space $V$ to the $\DN(\C)$-span of $ v + \eps v \times p  $, is not even surjective. Indeed, if $v,w \in V$ are such that $v \neq 0$, $\langle v,v \rangle = 0$, and $\langle v,w \rangle \neq 0$, then the $\DN(\F)$-span of $v+ \eps w$ contains no vectors $v' + \eps w' $ with $\langle v',w' \rangle = 0$, and hence no vectors of the form $ v + \eps v \times p  $.
\end{remark}

\subsection{Projective spaces over path algebras}\label{sec:path}
In this section, we describe projective spaces over path algebras of quivers, with the algebra of triangular matrices as the main example. We will see that these can be understood in terms of flags and Schubert cells in the Grassmannian. For a detailed account of quivers and their path algebras, see, e.g., \cite{kirillov2016quiver, schiffler2014quiver}.

Let $Q$ be a quiver, i.e., a directed graph, possibly with multiple edges and loops. Fix some ground field $\F$.

\begin{definition}
The \emph{path algebra} $\F Q$ of $Q$ is an  $\F$-algebra spanned, as an $\F$-vector space, by directed paths in $Q$, including, for every vertex $i$ of $Q$, a trivial path $e_i$, starting and ending at $i$. Multiplication in $\F Q$ corresponds to concatenation of paths: $ab$ is a path obtained by following first $b$ and then $a$. Concatenation $ab$ is only defined when the endpoint of $b$ coincides with the starting point of $a$. If concatenation is undefined, then the product of the corresponding paths is $0$.
\end{definition}
For any quiver $Q$, its path algebra $\F Q$ is unital and associative. The identity element is given by $1 = \sum e_i$, where the sum is taken over all vertices $i$ of $Q$. The path algebra is finite-dimensional if and only if the quiver has no directed cycles.
\begin{example}\label{ex:Aquiver}
Consider a Dynkin quiver of type $A_k$, with all arrows pointing in one direction:
$$
1\to 2 \to \dots \to k.
$$
For every $i \leq j$, where $1 \leq i,j \leq k$, this quiver has a unique directed path from $i$ to $j$. Denote it by $E_{ji}$. Then the corresponding path algebra is spanned by $E_{ij}$, $i \geq j$, with multiplication given by
$
E_{ij}E_{jl} = E_{il}
$
and all other products being $0$. 
This is precisely the algebra $\mathrm T_k(\F)$ of $k \times k$ lower-triangular matrices over $\F$, with $E_{ij}$ being matrix units, cf. \cite[Example 4.3]{schiffler2014quiver} (note that the algebras of upper- and lower-triangular matrices are isomorphic; we choose to work with lower-triangular matrices to simplify notation).
\end{example}

Modules over path algebras can be understood in terms of representations of the corresponding quiver.

\begin{definition}
    A \emph{representation} of a quiver $Q$ is an assignment of a vector space $V_i$ to its every vertex $i$ and a linear map $\phi_a \colon V_i \to V_j$ to every directed edge $a$ from $i$ to $j$.
\end{definition}
Subrepresentations, isomorphisms, direct sums, etc. of quiver representations are defined in an obvious way.

There is a natural identification between quiver representations and (with our conventions, left) modules over the corresponding path algebra, see \cite[Theorem 1.7]{kirillov2016quiver}. Given a $\F Q$-module $M$, set $V_i:= e_iM$; these are the vector spaces assigned to the vertices of $Q$. The map $\phi_a \colon V_i \to V_j$ assigned to a directed arrow $a$ from $i$ to $j$ is given by the action of $a$, viewed as an element of $\F Q$. We have $aV_i \subset V_j$ since $ae_i = e_ja$.

Conversely, given a representation of a quiver $Q$, the corresponding $\F Q$-module is defined as the direct sum of vector spaces $V_i$ assigned to the vertices of $Q$. The action of $\F Q$ on $\bigoplus V_i$ is defined as follows: if $a$ is an arrow from vertex $i$ to vertex $j$, then the action of $a$ is given by $av = \phi_a(v)$ if $v \in V_i$, and $av = 0$ otherwise.

\begin{example}\label{ex:aquiverrep}
A representation of the quiver from Example \ref{ex:Aquiver} is a sequence of vector spaces and maps
$$
\begin{tikzcd}
V_1 \arrow[r, "\phi_{1}"]   &  V_2 \arrow[r, "\phi_{2}"] & \dots\arrow[r, "\phi_{k-2}"] & V_{k-1 }\arrow[r, "\phi_{k-1}"] & V_k
\end{tikzcd}
$$
The associated module over the algebra $\mathrm T_k(\F)$ of lower-triangular matrices is the vector space $V_1 \oplus \dots \oplus V_k$. The matrix $E_{ii}$ acts as a projector onto $V_i$. The matrix $E_{i+1, i}$ acts as $\phi_i$ on $V_i$ and trivially on other $V_j$.

The natural $\mathrm T_k(\F)$-module structure on $\F^k$ corresponds to all $V_i$ being one-dimensional and all $\phi_i$ being isomorphisms, while $\mathrm T_k(\F)$ as a module over itself corresponds to $\dim V_i = i$, and all $\phi_i$ being injective.

\end{example}

This interpretation of modules over path algebras makes the description of corresponding projective spaces a relatively straightforward task. From now on, we will concentrate on the quiver from Example \ref{ex:Aquiver} corresponding to the algebra $\mathrm T_k(\F)$ of triangular matrices. We will see that the corresponding projective spaces can be understood in terms of flags. More general quivers lead to configurations of subspaces with more complicated combinatorics, see Example \ref{ex:quiver2} below.

Consider a representation of the quiver from Example \ref{ex:Aquiver}, as described in Example~\ref{ex:aquiverrep}. Assuming that all maps $\phi_i$ are injective (which is the case for free modules), we can think of such a representation as a flag $V_1 \subset  \dots \subset V_k$. If we further assume that all $V_i$ are finite-dimensional (again, that is so for free modules), then such a flag is determined, up to isomorphism, by its \emph{type}, i.e., a non-decreasing sequence of integers $\dim V_1\leq  \dots \leq \dim V_k$. A free rank one module over $\mathrm T_k(\F)$ corresponds to a complete flag, i.e., a flag of type $(1, \dots, k)$, see Example \ref{ex:aquiverrep}. So, given a flag, the projectivization of the corresponding $\mathrm T_k(\F)$ module consists of its complete subflags which admit a direct complement. These are described by the following:

\begin{prop}
Suppose $V_1 \subset \dots \subset V_k$ is a flag and $U_1 \subset \dots \subset U_k$ is its subflag, i.e., $U_i \subset V_i$. Then the following are equivalent:
\begin{enumerate}
\item The subflag $U_1 \subset \dots \subset U_k$ admits a direct complement in $V_1 \subset \dots \subset V_k$, i.e., there is another subflag $W_1 \subset \dots \subset W_k$ such that $V_i = U_i \oplus W_i$ for all  $i = 1, \dots, k$.
\item $U_i = V_i \cap U_{i+1}$ for all $i = 1, \dots, k-1$.
\item $U_i = V_i \cap U_k$ for all $i = 1, \dots, k$.
\end{enumerate}
\end{prop}
\begin{proof}
{$ 1 \implies 2$.} We have $U_i \subset V_i \cap U_{i+1}$  since $U_1 \subset \dots \subset U_k$ is a subflag of $V_1 \subset \dots \subset V_k$. To show that $U_i \supset V_i \cap U_{i+1}$, suppose $v \in V_i \cap U_{i+1}$. Since $v \in V_i$ and $V_i = U_i \oplus W_i$, we have $v = u + w$ where $u \in U_i$, $w \in W_i$. Since $u\in U_i$, we also have $u\in U_{i+1}$. Since $v \in U_{i+1}$, it follows that $w \in U_{i+1}$. On other hand $w \in W_i \subset W_{i+1}$, and $U_{i+1} \cap W_{i+1} = 0$, so $w =0$, and $v = u \in U_i$, as needed.

\smallskip

{$ 2 \implies 1$.} We prove the result by induction on $k$. The statement is vacuously true for $k=1$. For general $k$, by induction hypothesis, the subflag $U_1 \subset \dots \subset U_{k-1}$ of $V_1 \subset \dots \subset V_{k-1}$ admits a complementary flag $W_1 \subset \dots \subset W_{k-1}$. To extend it to a flag $W_1 \subset \dots \subset W_{k}$ complementary to $U_1 \subset \dots \subset U_{k}$, we need to find a subspace $W_k \supset W_{k-1}$ such that $V_k = U_k \oplus W_k$. Clearly, such a subspace exists as long as $U_k \cap W_{k-1} = 0$. To see the latter, observe that
$
U_k \cap W_{k-1} \subset U_k \cap V_{k-1} = U_{k-1},
$
which has a trivial intersection with $ W_{k-1}$. Thus, $U_k \cap W_{k-1} = 0$.
\smallskip

$ 2 \implies 3$. Assume that  $U_i = V_i \cap U_{i+1}$ for all $i = 1, \dots, k-1$. Then, for any such $i$, we have 
$$
U_i = V_i \cap U_{i+1} = V_i \cap (V_{i+1} \cap U_{i+2} ) = \dots = V_i \cap \dots \cap V_k \cap U_k = V_i \cap U_k. 
$$
Also, $U_i = V_i \cap U_k$ is trivially true for $i=k$. The result follows.

\smallskip

$ 3 \implies 2$.  Assume that $U_i = V_i \cap U_k$. Then $$U_i = U_i \cap U_{i+1} = (V_i \cap U_k) \cap U_{i+1} = V_i \cap (U_k \cap U_{i+1}) = V_i \cap U_{i+1}. $$
\qedhere
\end{proof}

So, subflags of $V_1 \subset \dots \subset V_k$ admitting a direct complement are those of the form $U \cap V_i$, where $U = U_k$ is a subspace of $V_k$. In particular, we obtain the following:
\begin{prop}

Points in the projectivization of a $\mathrm T_k(\F)$-module associated with a flag $V_1 \subset \dots \subset V_k$ correspond to subspaces $U \in \Gr(k,V_k)$ such that
$
\dim(U \cap V_i) = i
$ for all $i=1,\dots,k$, cf. Proposition \ref{prop:prgr}.
\end{prop}

It is easy to see that the set of such $U$ is a partial closure of a Schubert cell in $\Gr(k,V_k)$ corresponding to a partition $\lambda_i := (\dim V_k - k) - (\dim V_i - i) $, provided that the sequence $\lambda_i$ is non-increasing, i.e., $\dim V_{i+1} > \dim V_i $ for all $i$. If this condition does not hold, then a complete subflag of $V_1 \subset \dots \subset V_k$ cannot admit a direct complement, so the projectivization of the corresponding module is empty.

Further, since free rank one $\mathrm T_k(\F)$-modules correspond to complete flags, a free rank $n$ module corresponds to a flag of type $(n, 2n, \dots , kn)$. In particular, we have the following:

\begin{prop}
The projective plane $\Proj^2(\mathrm T_k(\F))$ consists of  $P \in \Gr(k,3k)$ such that
$
\dim(P \cap \F^{3i}) = i,
$
where $\F^3 \subset \dots \subset \F^{3n}$ is a flag of coordinates subspaces. Lines in that projective plane are given by $L \in \Gr(2k,3k)$ such that
$
\dim(L \cap \F^i) = 2i.
$
Incidence between points and lines is given by containment, cf. Proposition \ref{prop:matrixplane}.
\end{prop}
For instance, points in $\Proj^2(\mathrm T_2(\F))$ are $2$-planes in $\F^6$ which meet a fixed $3$-plane along a line, while lines are $4$-planes which meet the same $3$-plane along a $2$-plane.

In a similar way, one can describe projective spaces for path algebras of more complicated quivers. In that case, instead of flags one needs to study more complex configurations of subspaces.
\begin{example}\label{ex:quiver2}
Modules over the path algebra of the quiver
$
1\to 3 \leftarrow  2
$
correspond to triples of vector spaces $V_1, V_2, V_3$ with linear maps $\phi_1 \colon V_1 \to V_3$, $\phi_2 \colon V_2 \to V_3$. Assuming that $\phi_1, \phi_2$ are injective and their images have trivial intersection (which is the case for free modules), this is the same as a vector space $V_3$ with two trivially intersecting subspaces $V_1$, $V_2$.
For a free rank one module one has $\dim V_1 = \dim V_2 = 1$, $\dim V_3 = 3$. Using this, it is easy to see that the projective plane over this algebra consists of $3$-planes $U \in \Gr(3,9)$ which meet each of the two given  $3$-planes (which are transverse to each other) along a line.

\end{example}

\section{Polygons and pentagram maps over rings}\label{sec:ppm}
In this section, we define the notion of a \emph{polygon} in the projective plane over a ring (see Definition \ref{defpolygon}), as well as show that the pentagram map on polygons is well-defined (Section~\ref{GeometryProjPlane}). Furthermore, we show that, for matrix rings and dual numbers, the general notion of a pentagram map over a ring specializes to Grassmann and skewer pentagram maps respectively (Examples \ref{ex1} and \ref{ex2}). 
\subsection{Polygons in projective spaces over rings} \label{polygons}

Over a field, the pentagram map is well-defined for polygons such that all consecutive triple{s} of vertices are in general linear position. The general linear position requirement means that the points are not collinear. In the context of (stably finite) rings, the right condition that ensures that the pentagram map is defined is the \emph{spanning} condition.
\begin{definition}
Let $M$ be an $R$-module. We say that points $P_1, \dots, P_n \in \Proj(M)$ \emph{span} $ \Proj(M)$ if their corresponding subspaces span $M$: $P_1 + \dots + P_n = M$. 
\end{definition}
By Proposition \ref{prop:gsb}, if $R$ is stably finite, and $M$ is free of rank $n$, then $n$ points $P_1, \dots, P_n$ span $ \Proj(M)$ if and only if $M = P_1 \oplus \dots \oplus P_n$. In other words, if $v_i$ spans $P_i$, then  $P_1, \dots, P_n$ span $ \Proj(M)$ if and only if $v_1, \dots, v_n$ is a basis of $M$. 

\begin{definition} \label{defpolygon}
A \textit{polygon} in $\Proj^d(R)$ is a bi-infinite sequence of points $P_i \in \Proj^d(R)$ such that, for any $i\in \Z$, the points $P_i, P_{i+1}, \dots, P_{i+d}$ span $\Proj^d(R)$.

\end{definition}
For stably finite $R$, this definition can be restated as follows: if $\Proj^d(R) = \Proj(M)$, and $v_i$ spans $P_i$, then $v_i, v_{i+1}, \dots, v_{i+d}$ is a basis of $M$  for any $i\in \Z$.

We notice that our definition of polygons allows for infinite polygons or closed ones, that is, polygons such that $P_i = P_{i+n}$ for some fixed $n \in \Z_+$ and all $i \in \Z$.

\begin{definition} \label{defProjEq}
We say two polygons $(P_i)$ and $( P_i')$ are  \emph{projectively equivalent} if there exists a  projective transformation that takes each $P_i$ to $ P_i'$.
\end{definition}
For most of the paper, we are interested in polygons considered up to projective equivalence. 

\begin{remark}
We note that the notion of being spanning, and hence the notion of a polygon, is not determined by the incidence structure. For example, it is easy to see that the inclusion map $\Z^3 \to \mathbb Q^3$ gives an isomorphism of incidence spaces $\Proj^2(\Z) \simeq \Proj^2(\mathbb Q)$ (the same is true for any PID and its field of fractions). 

However, while any three non-collinear points $P_1, P_2, P_3$ span $\Proj^2(\mathbb Q)$, for $\Proj^2(\Z)$ that is only so if the determinant of the matrix formed by the basis vectors in  $P_1, P_2, P_3$ is $\pm 1$. Accordingly, not every polygon over $\mathbb Q$ is a polygon over $\Z$ in the sense of Definition \ref{defpolygon}, even though $\Proj^2(\Z)\simeq\Proj^2(\mathbb Q)$ as incidence spaces.
\end{remark}

\subsection{Pentagram maps over stably finite rings, Grassmannians, and skewers} \label{GeometryProjPlane}

\begin{definition} \label{defGeneralizedPM}
Suppose $R$ is a stably finite ring. The \emph{pentagram map} in $\Proj^2(R)$ takes a polygon $P_i \in \Proj^2(R)$ to the sequence of points $P'_i \in \Proj^2(R)$, where $P_i'$ is the intersection of the lines connecting $P_i$ to $P_{i+2}$ and $P_{i+1}$ to $P_{i+3}$.  
\end{definition}

The sequence $P_i'$ is well-defined thanks to the following. 

\begin{lemma} \label{uniqueintpt}
Suppose $R$ is a stably finite ring, and $M$ is a free rank three $R$-module. Let $P_1, P_2, P_3, P_4 \in  \Proj(M)$ be such that the triples $P_1, P_2, P_3$ and $P_2, P_3, P_4$  span $\Proj(M)$. Then

\begin{enumerate}
\item For each of the pairs of points $(P_1, P_2)$, $(P_3, P_4)$, there is a unique line containing that pair of points.
\item Those lines $P_1P_2$ and $P_3P_4$ meet at a single point. 

\end{enumerate}
\end{lemma}

\begin{proof}
To prove the first claim, note that since $P_1, P_2, P_3$ span $\Proj(M)$, we have that $P_1 \oplus P_2$ is free, and $(P_1 \oplus P_2) \oplus P_3 = M$, so $P_1 \oplus P_2 \subset M$ is a rank two subspace (i.e., a line in $\Proj(M)$). Furthermore, any other line containing $P_1, P_2$ must contain $P_1 \oplus P_2$ and thus coincide with it by Corollary \ref{cor:di}. So, $P_1, P_2$ lie on a unique line, $P_1 \oplus P_2$. Likewise, $P_3, P_4$ also belong to a unique line, $P_3 \oplus P_4$. 

To prove the second claim, it suffices to show that $(P_1 \oplus P_2) \cap (P_3 \oplus P_4)$ is a rank one subspace (i.e., a point in $\Proj(M)$). If this is done, then, by Corollary \ref{cor:di}, any other point contained in the lines $P_1 \oplus P_2$, $P_3 \oplus P_4$ must coincide with  $(P_1 \oplus P_2) \cap (P_3 \oplus P_4)$. 

First, observe that, since $M =  (P_3 \oplus P_4) \oplus P_2$, and $P_2 \subset P_1 \oplus P_2$, we have \begin{equation}\label{eqn:p12}P_1 \oplus P_2  =  ((P_1 \oplus P_2) \cap (P_3 \oplus P_4)) \oplus P_2.\end{equation}
So, there is an isomorphism of submodules $(P_1 \oplus P_2) \cap (P_3 \oplus P_4) \simeq P_1$, given by projection along $P_2$. In particular, $(P_1 \oplus P_2) \cap (P_3 \oplus P_4)$ is a rank one free module. Also, using \eqref{eqn:p12}, we get
$$
M = P_1 \oplus P_2 \oplus P_3  =  ((P_1 \oplus P_2) \cap (P_3 \oplus P_4)) \oplus P_2 \oplus P_3.
$$

So, the rank one free submodule $(P_1 \oplus P_2) \cap (P_3 \oplus P_4)$ admits a direct complement in $M$, i.e., is a point. \qedhere

\end{proof}

Applying Lemma \ref{uniqueintpt} to vertices $P_i, P_{i+2}, P_{i+1}, P_{i+3}$  (in that specific order) of a polygon~$P_i$, we get that the pentagram map on polygons is well-defined. We note that the pentagram map is \emph{not} a map from the set of polygons to itself, since the image of a polygon under the pentagram map is not necessarily a polygon. Instead, it should be understood as a binary relation on the set of polygons. In the case when the base ring is a field, the pentagram map is, actually, given by rational formulas \ai{(see, e.g., \cite[Section 2.2]{integrabilityPMOvScTa})}.

Applying Lemma \ref{uniqueintpt} to vertices $P_i, P_{i+1}, P_{i+2}, P_{i+3}$ (in that order) of a polygon $P_i$, one also gets the following map:
\begin{definition} \label{defGeneralizeIPM}
The \emph{inverse pentagram map} in $\Proj^2(R)$ takes a polygon $P_i \in \Proj^2(R)$ to the sequence of points $P'_i \in \Proj^2(R)$, where $P_i'$ is the intersection of the lines connecting $P_i$ to $P_{i+1}$ and $P_{i+2}$ to $P_{i+3}$. 
\end{definition}
This is indeed the inverse of the pentagram map in the sense that the composition of this map and the pentagram map (taken in either order) is the identity, up to shifting indices, wherever that composition is defined. To see this, consider a polygon $P_i \in \Proj^2(R)$ that is mapped by the pentagram map to a polygon $P'_i \in \Proj^2(R)$. Then, by construction, we have that $P_i', P_{i+1}'$ belongs to the line $P_{i+1}P_{i+3}$. Since $P_i'$ is a polygon, it follows that $P_{i}'P_{i+1}' = P_{i+1}P_{i+3}$ and $P_{i}'P_{i+1}' \cap P_{i+2}'P_{i+3}' = P_{i+3}$. Thus, the pentagram map followed by the inverse pentagram is the identity (up to shifting indices). Likewise, if we apply the inverse map first, we get the identity as well.

We note that both the pentagram map and its inverse commute with projective transformations and, as such, are well-defined on the set of polygons modulo projective equivalence. Also, note that if the image of a closed polygon under the pentagram map is again a polygon, then it is also closed. 

\begin{example}\label{ex1}

The pentagram map over the matrix ring $\Mat_k(\F)$ is precisely (the planar case of) the Grassmann pentagram map, as defined in \cite[Definition 4.1]{felipe2015pentagram}:
\begin{prop}
A polygon in $\Proj^2(\Mat_k(\F))$ is a sequence of $k$-dimensional subspaces  $V_i\subset \F^{3k}$ such that $V_i \oplus V_{i+1} \oplus V_{i+2} = \F^{3k}$. The pentagram map sends such a sequence to the sequence \begin{equation}\label{eq:gpm}V_i' := (V_i \oplus V_{i+2}) \cap (V_{i+1} \oplus V_{i+3}). \end{equation}
\end{prop}
\begin{proof}
We have $\Proj^2(\Mat_k(\F))  = \Proj(V^k)$, where $V$ is a vector space over $\F$ of dimension $3k$. The correspondence $W \mapsto W^k$ identifies the lattice of subspaces of $V$ with the lattice of submodules of~$V^k$. Under this correspondence, points of $ \Proj(V^k)$ are identified with $k$-dimensional subspaces of~$V$. Three points span $ \Proj(V^k) $ if and only if the corresponding $k$-dimensional subspaces span~$V$. Thus, a polygon in $\Proj(V^k)$ is precisely a collection $V_i \in \Gr(k,V)$
 such  that  $V_i \oplus V_{i+1} \oplus V_{i+2} = V$. Further, intersections and spans of subspaces in $V^k$ correspond to intersections of subspaces in $V$, giving the claimed formula for the pentagram map. \qedhere
 
  \end{proof}
  Polygons and the pentagram map for triangular matrices have a similar description. Polygons in $\Proj^2(\mathrm T_k(\F))$  consist of sequences of subspaces $V_i \subset \F^{3k}$ such that, for all $i \in \Z$ and all $j = 1, \dots, k$, we have
 $$\dim(V_i \cap \F^{3j}) = j, \quad
  (V_i \cap \F^{3j}) \oplus (V_{i+1} \cap \F^{3j}) \oplus (V_{i+2} \cap \F^{3j}) = \F^{3j}.$$ The pentagram map on these sequences is given by the same formula as for the full matrix ring.

  For example, a polygon in   $\Proj^2(\mathrm T_2(\F))$ is a sequence of $2$-planes $V_i$ in $\F^6$ such that $V_i \oplus V_{i+1} \oplus V_{i+2} = \F^6$, and each $V_i$ intersects a given subspace $\F^3 \subset \F^6$ along a line $W_i$. Additionally, $W_i \oplus W_{i+1} \oplus W_{i+2} = \F^3$. The pentagram map produces new sequences $V_i', W_i'$ given by 
  \begin{gather}V_i' := (V_i \oplus V_{i+2}) \cap (V_{i+1} \oplus V_{i+3}),
  \\
  W_i' := (W_i \oplus W_{i+2}) \cap (W_{i+1} \oplus W_{i+3}).\end{gather}
 Notice that $W_i$, which are just points in $\Proj^2(\F)$, undergo  the usual pentagram map. Thus, the pentagram map in $\Proj^2(\mathrm T_2(\F))$ can be thought of as a Grassmannian pentagram map on sequences $V_i \in \Gr(2,6)$ coupled with the usual pentagram map on sequences $W_i \in \Gr(1,3)$. The coupling condition is $W_i = V_i \cap \F^3$. A similar description can be given for pentagram maps over $\mathrm T_k(\F)$ or, more generally, arbitrary path algebras corresponding to acyclic quivers.

  \end{example}
  
  \begin{example}\label{ex2}
  The pentagram map over the ring $\DN$ of dual numbers is the skewer pentagram map, as defined in \cite[Section 9.7]{tabwineskins}:
  \begin{prop}
  A polygon in $\Proj^2(\DN)$ is a bi-infinite sequence $\ell_i$ of affine lines in $\R^3$ such that, for any $i \in \Z$, the direction vectors of the lines $\ell_i, \ell_{i+1}, \ell_{i+2}$ are linearly independent. The pentagram map sends such a sequence to the sequence    $$\ell_i' := S(S(\ell_i , \ell_{i+2}) , S(\ell_{i+1} , \ell_{i+3})),$$ where $S(\ell,\ell')$ stands for the common perpendicular (skewer) of the lines $\ell,\ell'$. 
  \end{prop}
  \begin{proof}
  Consider three affine lines $\{p_1 + t v_1\}$, $\{p_2 + t v_2\}$, $\{p_3 + t v_3\}$ in $\R^3$. The associated points in $\Proj^2(\DN)$ are the spans of the dual vectors $\zeta(v_i,p_i) = v_i + \eps v_i \times p_i$. Those points span $\Proj^2(\DN)$ if and only if the vectors $\zeta(v_i,p_i)$ span $\DN^3$. Since $\DN$ is a commutative ring, this is equivalent to saying that  $\det(\zeta(v_1,p_1), \zeta(v_2,p_2), \zeta(v_3,p_3)) $ is a unit. 
 Consider the real part function $\Re \colon \DN \to \R$, $\Re(a+b \eps):= a$. It is a homomorphism of rings, so
$$
\Re  \det(\zeta(v_1,p_1), \zeta(v_2,p_2), \zeta(v_3,p_3)) =  \det(\Re\zeta(v_1,p_1), \Re\zeta(v_2,p_2), \Re \zeta(v_3,p_3)) =  \det(v_1, v_2, v_3).
$$
 Therefore, three affine lines span $\Proj^2(\DN)$ if and only if their direction vectors span $\R^3$, giving the desired description of polygons.  

  As for the formula for the pentagram map, it follows from the definition of the latter and Proposition \ref{prop:iis}. 
  \end{proof}

  \end{example}

\section{Polygons via difference operators} \label{DO}

In this section, we introduce our main technical tool: difference operators with coefficients in a ring. After discussing basic properties of such operators in Section \ref{DODefs}, we use that language to give an algebraic description of the moduli space of polygons in the projective plane over a ring (Theorem \ref{DOcorresPOLY}). 

\subsection{Difference operators with coefficients in a ring} \label{DODefs}
For a ring $R$, denote by $R^\Z$ the ring of bi-infinite sequences of elements of $R$. We will think of $R^\Z$ as a left module over itself and, simultaneously, as a right $R$-module. Denote by $\End_R(R^\Z)$ the ring of endomorphisms of $R^\Z$, viewed as a right $R$-module.

\begin{definition} \label{defDifferenceOperators}
The ring of (left) \emph{difference operators over $R$} is a subring $R^\Z[T]\subset \End_R(R^\Z)$ generated by operators of left multiplication by elements of $R^\Z$, along with the \emph{shift} operator  $T \in \End_R(R^\Z)$, given by $(T{a})_i := a_{i+1}$. In other words, a difference operator is an endomorphism $\D : R^\Z \to R^\Z$ \ai{of right $R$-modules} given by \[\D = \sum_{i=0}^d \alpha_{i} T^i,\]
where $d \geq 0$ is an integer,  and $\alpha_1, \dots, \alpha_d \in R^\Z$. Assuming $\alpha_d \neq 0$, we call $d$ \emph{the degree} of the difference operator $\D$. We will denote the entries of the bi-infinite sequence $\alpha_i$ by $\alpha_{ij}$. For any $a \in R^\Z$, the operator $\D$ acts on $a$ as
$$
(\D a)_j=  \sum_{i=0}^d \alpha_{ij} a_{i+j}.
$$
 
\end{definition}
More generally, one can consider the subring   $R^\Z[T,T^{-1}] \subset \End_R(R^\Z)$ generated by operators of left multiplication by elements of $R^\Z$, along with $T$ and $T^{-1}$, that is, allow for difference operators of the form  \[\D = \sum_{i=m}^n \alpha_i T^i,\] where $m,n \in \Z$. However, we will focus primarily on operators as in Definition \ref{defDifferenceOperators}.

\begin{definition}
Let $\D = \sum_{i=0}^d \alpha_i T^i \in R^\Z[T]$ be a difference operator. 
If $\alpha_0$ and $\alpha_d$ are sequences of left invertible elements in $R$ (i.e., are left invertible in $R^\Z$), we call $\D$ a \textit{properly bounded difference operator}.  
\end{definition}

In the next section, we will establish a correspondence between properly bounded difference operators and polygons. The remaining part of the current section contains various technical results on difference operators which will be used to establish that correspondence. 

To a properly difference operator $\D$ of degree $d$, we will associate a polygon in the projectivization of $(\Ker \D)^*$. The latter is a $(d-1)$-dimensional left projective space, thanks to the following:

\begin{prop} \label{propKernelDim}
    Let $\mathcal{D} \in R^\Z[T]\subset \End_R(R^\Z)$ be a properly bounded difference operator of degree $d$. Then its kernel $\Ker\, \mathcal{D} \subset R^\Z$ is a free rank $d$ right $R$-module. 
\end{prop}
\begin{proof}
Since $\D$ is an endomorphism of the right $R$-module $R^\Z$, its kernel is also a right $R$-module. To prove that it is free of rank $d$, notice that the equation $\D a = 0$ may be viewed as a sequence of recurrence relations, allowing one to find all entries of $a \in R^\Z$ from $a_1, \dots, a_d$. Thus, the projection mapping $\Ker\, \mathcal{D} \to R^d$, given by $a \mapsto (a_1, \dots, a_d)$, is an isomorphism of right $R$-modules, and $\Ker\, \mathcal{D}$ is free of rank $d$.
\end{proof}

For a difference operator $\mathcal{D} \in R^\Z[T]$ which is not properly bounded,  the kernel does not have to be free and/or of finite rank. Still, the kernel of such an operator is relatively small in the following sense:
\begin{prop}\label{prop:kernpb}
 Let $\mathcal{D} \in R^\Z[T]$, $\D \neq 0$ be a difference operator of degree less than $d$. Then, for any properly bounded difference operator $\mathcal{D}' \in R^\Z[T]$ of degree $ d$, there exists $a \in \Ker \D'$ such that $a \notin \Ker \D$. 
\end{prop}
\begin{proof}
Since $\D \neq  0$, there is $i \in \Z$ such that the condition $\D a = 0$, where $a \in R^\Z$, imposes a non-trivial relation on $d$ consecutive entries $a_i, \dots, a_{i+d-1}$ of $a$. At the same time, the map $\Ker \D' \to R^d$ given by $a \mapsto (a_i, \dots, a_{i+d-1})$ is a bijection, so we can find an element of $\Ker \D' $ whose entries do not satisfy the relation forced by $\D a = 0$.
\end{proof}

Now, let $M$ be a free finite-rank $R$-module. Then $R^\Z \otimes_R M \simeq M^\Z$, as left $R^\Z$-modules (where $M^\Z$ consists of bi-infinite sequences valued in $M$). Any difference operator $\mathcal{D} \in R^\Z[T]\subset \End_R(R^\Z)$ acts on $M^\Z = R^\Z \otimes_R M$ by $\D(a \otimes x) = (\D a) \otimes x$, where $a \in R^\Z$, $x \in M$. Put differently, for $\D = \sum_{i=0}^d \alpha_i T^i$ and $v = (v_i)_{i\in \Z} \in M^\Z$, we have
$$
(\D v)_j =  \sum_{i=0}^d \alpha_{ij} v_{i+j}.
$$

\ai{Let $e_1, \dots, e_m \in M$ be a basis.} For
$$
v = \sum_{i=1}^m a_i \otimes e_i \in R^\Z \otimes_R M = M^\Z, 
$$
define a submodule $\mathrm{Span}\, v $ of the right $R$-module $ R^\Z$ as the right span of $a_1, \dots, a_m$. In other words, choosing a basis in $M$ allows one to write $v \in M^\Z$ as an infinite matrix
\begin{equation}\label{vmatrix}\left(\begin{array}{ccc}  a_1 & \dots & a_m\end{array}\right)=\left(\begin{array}{ccc} \vdots && \vdots  \\ v_{i,1} & \dots & v_{i,m} \\ \vdots && \vdots\end{array}\right)\end{equation}
 with columns indexed by basis elements and rows indexed by $\Z$. The rows of this matrix are just the basis expansions of the entries of the sequence $v$. Then $\mathrm{Span}\, v$ is the column space of this matrix. Note that the submodule $\mathrm{Span}\, v$  does not depend on the choice of a basis in $M$, as it can be defined as the set bi-infinite sequences of the form $(\xi(v_i))_{i \in \Z}$, where $\xi \in M^*$.
\begin{prop}\label{prop:trinker}
Suppose $\D \in R^\Z[T]$ is a difference operator, and $v \in M^\Z$, where $M$ is a free finite-rank left $R$-module. Then $\D v = 0$ if and only if $\mathrm{Span}\, v \subset \Ker \,\D$.
\end{prop} 
\begin{proof}
Let $e_1, \dots, e_m$ be a basis in $M$. Then $v = \sum_{i=1}^m a_i \otimes e_i$ for some $a_i \in R^\Z$. The submodule $\mathrm{Span}\, v$ is spanned by the $a_i$'s. At the same time,
$$
\D( v ) =  \sum_{i=1}^m \D(a_i) \otimes e_i,
$$
which is $0$ if and only if $ \D(a_i) \ai{\,= 0}$ for all $i$, which is the same as to say $\mathrm{Span}\, v \subset \Ker \,\D$.
\end{proof}
The next proposition is one of the cornerstones of our construction and is where stable finiteness comes into play. For a sequence $v \in M^\Z$ such that $\D v = 0$, it gives a criterion for when the associated sequence of points in $\Proj(M)$ is a polygon. This criterion is a difference counterpart of the following standard result from the theory of ordinary differential equations: for functions $f_1(x), \dots, f_n(x)$ belonging to a kernel of the differential operator $$\Delta = (d/dx)^n + u_{n-1}(x)(d/dx)^{n-1} + \dots + u_0,$$ let $W(x) = \det(f_i^{(j)})$, where $i = 1, \dots, n$, $j = 0, \dots, n-1$, be their Wro\'nskian. Then $f_1(x), \dots, f_n(x)$ span the kernel of $\Delta$  $\iff$ $W(x) \neq 0$ for all $x$ $\iff$$W(x) \neq 0$ for some $x$. 

\begin{prop}\label{wronskian}
Let $R$ be a stably finite ring. Suppose $\D \in R^\Z[T]$ is a properly bounded difference operator of degree $d$, and $v \in M^\Z$, where $M$ is a free left $R$-module of rank $d$. Suppose that $\D v = 0$. Then the following conditions are equivalent:
\begin{enumerate}
\item $\mathrm{Span}\, v = \Ker \,\D$.
\item $v_k, \dots, v_{k+d-1}$ span $M$ (equivalently, form a basis in $M$) for some $k \in \Z$.
\item $v_k, \dots, v_{k+d-1}$ span $M$ (equivalently, form a basis in $M$) for all $k \in \Z$.
\end{enumerate}

\end{prop} 
\begin{proof}
Let $k \in \Z$. It suffices to show that $v_k, \dots, v_{k+d-1}$ span $M$ if and only if $\mathrm{Span}\, v = \Ker \,\D$. Take a basis $e_1, \dots, e_d$ in $M$. 
Consider a $d \times d$ matrix $$ V_k:= \left(\begin{array}{ccc}v_{k,1} & \dots & v_{k,d} \\ \vdots && \vdots \\v_{k+d-1,1} & \dots & v_{k+d-1,d}\end{array}\right) \in \Mat_d(R)$$ where $v_{i,j} \in R$ are coordinates of $v_i$ relative to the basis  $e_1, \dots, e_d$. The rows of the matrix $V_k$ are coordinate expressions for $v_k, \dots, v_{k+d-1} \in M$, so $V_k$ is left invertible if and only if $v_k, \dots, v_{k+d-1}$ span $M$. At the same time, the matrix $V_k$ consists of $d$ consecutive rows of the infinite matrix \eqref{vmatrix} whose columns span the module $\mathrm{Span}\, v$. Since taking $d$ consecutive entries gives an isomorphism $\Ker \,\D \simeq R^d$, the columns of the matrix $V_k$ span $R^d$ if and only if the columns of the matrix  \eqref{vmatrix} span $\Ker \,\D$, which is equivalent to saying that $\mathrm{Span}\, v = \Ker \,\D$. To sum up, $v_k, \dots, v_{k+d-1}$ span $M$ if and only if the matrix $V_k$ is left-invertible, while $\mathrm{Span}\, v = \Ker \,\D$ if and only  if $V_k$ is right-invertible. Since $R$ is stably finite, those conditions are equivalent. 
\end{proof}

\subsection{The correspondence between difference operators and polygons}\label{sec:cor}
We will now establish a correspondence between polygons and difference operators. 

Note that the ring $R^\Z[T]$ of difference operators over $R$ contains a subring $R^\Z$ of bi-infinite sequences valued in $R$: such sequences are just degree zero difference operators. In particular, the group of units $(R^\Z)^\times$ in  $R^\Z$ (which consists of bi-infinite sequences of units in $R$) acts on $R^\Z[T]$ in two different ways: by left multiplication, and by right multiplication.

\begin{theorem} \label{DOcorresPOLY} 
Let $R$ be a stably finite ring. Then, there is a natural one-to-one correspondence between the following sets:
\begin{enumerate}
    \item Projective equivalence classes of polygons in the left  projective space $\Proj^{d-1}(R)$.
    \item Properly bounded left difference operators of degree $d$ with coefficients in $R$, up to multiplication on the left and right by bi-infinite sequences of units in $R$.
\end{enumerate}
\end{theorem}

By the set of projective equivalence classes of polygons in the left  projective space $\Proj^{d-1}(R)$, we mean the set of polygons in projectivizations of \emph{all} rank $d$ free left $R$-modules, considered up to projective equivalence. 

Similarly,  there is a natural one-to-one correspondence between the set of projective equivalence classes of polygons in the \emph{right}  projective space $\Proj^{d-1}(R)$, and the set of
   properly bounded \emph{right} difference operators of degree $d$ with coefficients in $R$, up to multiplication on the left and right by bi-infinite sequences of units in $R$. 
   
   In the case when the base ring $R$ is a field, the correspondence between polygons and difference operators is well-known. In particular, the case $R = \R$ follows from \cite[Proposition 3.3]{izopentagramandrefactorization}. 
   
   In future sections, we are most interested in the case of $d=3$, i.e.,  a correspondence between polygons in the projective plane $\Proj^2(R)$ and degree $3$ difference operators. Still, we give a proof in the general case because it works in the same way for any $d$. Also, we expect $d > 3$ cases of Theorem \ref{DOcorresPOLY} to be useful for the study of pentagram maps in higher-dimensional projective spaces, cf. \cite{izopentagramandrefactorization}.

In order to prove Theorem \ref{DOcorresPOLY}, we will first build a map from the space of all degree $d$ properly bounded left difference operators with coefficients in a stably finite ring $R$ to the space of polygons, up to projective equivalence, in the left projective space $\Proj^{d-1}(R)$. Then, we will consider difference operators up to the action \ai{of} bi-infinite sequences of units on left and right and show that our original map factors through this quotient. Finally, we will prove that the map on the quotient is a bijection.

Let us denote the space of all degree $d$ properly bounded difference operators over $R$ as $R^\Z[T]_d^{\mathrm{pb}}$ and the space of polygons in $\Proj^{d-1}(R)$ up to projective equivalence as $\mathfrak P(\Proj^{d-1}(R))$. Our first aim to build a map $\bar \varphi :R^\Z[T]_d^{\mathrm{pb}} \to \mathfrak P(\Proj^{d-1}(R))$.

Let $\D \in R^\Z[T]_d^{\mathrm{pb}}$. Then, by Proposition \ref{propKernelDim}, we know that $\K \subset R^\Z$ is a rank $d$ free right $R$-module. Thus,  $(\Ker{\D})^* = \Hom_R(\Ker{\D}, R)$ is a free rank $d$ left $R$-module. For any $i \in \Z$, let $\eps_{\D, i} \in (\Ker{\D})^*$ be defined by $\eps_{\D, i}(a) := a_i$ for any $a \in \Ker{\D}$. Then $\eps_{\D} \in ((\Ker{\D})^*)^\Z = R^\Z \otimes_R (\Ker{\D})^*$.

\begin{prop}\label{prop:spaneps}
For stably finite ring $R$, we have $\mathrm{Span}\, \eps_{\D} = \Ker \D$. In particular, $\D \eps_{\D} = 0$ (recall that a difference operator $\D = \sum_{i=0}^d \alpha_i T^i$ acts on $\eps_\D$ by
$
(\D \eps_\D)_j =  \sum_{i=0}^d \alpha_{ij} \eps_{\D,i+j}
$).
\end{prop}
\begin{proof}
The module $\mathrm{Span}\, \eps_{\D}$ consists of sequences of the form $a(\eps_{\D, i})$, where $a \in ((\Ker{\D})^*)^* = \Ker \D$. But $a(\eps_{\D, i}) = \eps_{\D, i}(a) = a_i$, so $\mathrm{Span}\, \eps_{\D}$ is precisely $\Ker \D$. In particular, by Proposition~\ref{prop:trinker}, we have $\D \eps_{\D} = 0$.\qedhere

\end{proof}

\begin{cor} \label{alphabasis}
For stably finite ring $R$ and for any $i \in \Z$, the list $\eps_{\D, i}, \dots, \eps_{\D,i+d-1}$ is a basis of $(\K)^*$.
\end{cor}
\begin{proof}
Apply Proposition \ref{wronskian}. 
\end{proof}

It follows that the sequence of points $P_i = R \eps_{\D, i}$ is a polygon in $\Proj((\K)^*)$. Thus, every properly bounded difference operator gives rise to a polygon in the projectivization of the dual to its kernel. 
This gives a map $\bar \varphi: R^\Z[T]_d^{\mathrm{pb}} \to \mathfrak P(\Proj^{d-1}(R))$. 
Denote the collection of properly bounded, degree $d$ left difference operators up to multiplication on the left and right by bi-infinite sequences of units as ${R^\Z[T]_d^{\mathrm{pb}}}/((R^\Z)^{\times} \times (R^\Z)^{\times})$.

\begin{prop}
For stably finite ring $R$, the map $\bar \varphi: R^\Z[T]_d^{\mathrm{pb}} \to \mathfrak P(\Proj^{d-1}(R))$ factors through a map $ \varphi: {R^\Z[T]_d^{\mathrm{pb}}}/((R^\Z)^{\times} \times (R^\Z)^{\times}) \to \mathfrak P(\Proj^{d-1}(R))$.
\end{prop}

\begin{proof}
We need to check that if we take a properly bounded left difference operator $\D$, and multiply it on the left or right by a bi-infinite sequence of units in $R$, then the image of the resulting operator under $\bar \varphi$ is in the same projective equivalence class as the image of $\D$. 
Consider $\D \in R^\Z[T]_d^{\mathrm{pb}}$ and a bi-infinite sequence $c \in (R^\Z)^{\times}$ of units in $R$. Note that multiplication on the left $\D \mapsto c\D$ does not change the kernel and hence the image under $\bar \varphi$. So, it suffices to establish projective equivalence of the polygons $\bar \varphi(\D c)$ and $\bar \varphi(\D)$. Clearly, left multiplication by $c$ gives a right $R$-module isomorphism $\lambda_c \colon \Ker \D c \to  \Ker \D $. 
Consider the associated left module isomorphism $\lambda_c^* \colon (\Ker \D c)^*\to (\Ker \D)^*$. As before, let $\eps_{\D c, i} \in (\Ker \D c)^*$ be defined by $\eps_{\D c, i}(a) = a_i$, and let $a \in \Ker \D$. Then
$$
\lambda_c^*(\eps_{\D c, i})(a) = \eps_{\D c, i}(\lambda_c(a)) = \eps_{\D c, i}(ca) = c_i a_i = c_i \eps_{\D, i}(a).
$$
So,  $\lambda_c^*$ maps $\eps_{\D c, i} \in (\Ker \D c)^*$ to a scalar multiple of $\eps_{\D, i} \in (\Ker \D)^*$. Therefore, the associated projective transformation $\Proj((\Ker \D c)^*) \to \Proj((\Ker \D)^*)$ maps   $R\eps_{\D c, i}$ to $R\eps_{\D, i}$, as needed.\qedhere

\end{proof}

To establish injectivity, we start with the following.

 \begin{lemma}\label{lemma:inj}
 Let $R$ be a stably finite ring. Suppose that $\D, \D' \in R^\Z[T]_d^{\mathrm{pb}}$ are properly bounded (left) difference operators of degree $d$. Assume also that $\Ker \D = \Ker \D'$. Then $\D' = c' \D$ for some $c' \in (R^\Z)^{\times}$.  
 \end{lemma}
 \begin{proof}
Since $\D' $ and $\D$ are properly bounded of the same degree $d$, there is a sequence of units $c' \in (R^\Z)^{\times}$ such that
$
\Rem := \D' - c' \D
$
has degree at most $d-1$. Note that since $\Ker \D' = \Ker \D$, we have $\Ker \Rem \supset  \Ker \D$. So, by Proposition \ref{prop:kernpb}, it follows that $\Rem = 0$, meaning that $\D' = c' \D$, as desired. \qedhere
 \end{proof}

\begin{prop}\label{prop:inj} For stably finite ring $R$, the map
 $ \varphi: {R^\Z[T]_d^{\mathrm{pb}}}/((R^\Z)^{\times} \times (R^\Z)^{\times}) \to \mathfrak P(\Proj^{d-1}(R))$
 is injective.
 \end{prop}
\begin{proof}
Consider two properly bounded left difference operators $\D, \D' \in R^\Z[T]_d^{\mathrm{pb}}$ such that $\overline{\varphi}(\D) = \overline{\varphi}(\D')$. We need to show that $\D$ and $\D'$ belong to the same class in the quotient ${R^\Z[T]_d^{\mathrm{pb}}}/((R^\Z)^{\times} \times (R^\Z)^{\times})$. 

The equality $\overline{\varphi}(\D) = \overline{\varphi}(\D')$ means that the polygons associated with $\D$, $\D'$ are projectively equivalent. That is, there is a left module isomorphism $\psi \colon (\Ker \D')^*\to (\Ker \D)^*$ such that $\psi(\eps_{\D', i}) = c_i \eps_{\D, i}$ for some sequence of units $c_i \in R$. Consider the dual $\psi^* \colon \Ker \D \to \Ker \D'$. For any $a \in \Ker \D$, we have
$$
\psi^*(a)_i = \eps_{\D', i}(\psi^*(a)) = \psi(\eps_{\D', i})(a) = c_i \eps_{\D, i}(a) = c_i a_i.
$$
So, $\psi^*$ is left multiplication by $c$, and $\Ker \D' = c \Ker \D$. Therefore, $\Ker \D' c = \Ker \D$, and, by \ai{Lemma \ref{lemma:inj}}, we have $\D'c = c' \D$, as needed. \qedhere

\end{proof}


Now, to complete the proof of Theorem \ref{DOcorresPOLY}, it suffices to show that the mapping  $\bar \varphi :R^\Z[T]_d^{\mathrm{pb}} \to \mathfrak P(\Proj^{d-1}(R))$ is surjective. In other words, given a polygon $(P_i)$ in $\Proj(M)$, where $M$ is a rank $d$ free left $R$-module, we need to find a difference operator $\D \in R^\Z[T]_d^{\mathrm{pb}}$ such that the associated polygon in $\Proj((\Ker \D)^*)$ is projectively equivalent to $(P_i)$.

\begin{prop}\label{prop:ri}
Let $R$ be stably finite. Suppose $v_i \in M$ are such that the sequence $P_i = Rv_i$ is a polygon. Then there exists a properly bounded degree $d$ difference operator $\D\ai{_v} \in R^\Z[T]_d^{\mathrm{pb}} $ such that $\D\ai{_v} v = 0$.
\end{prop}
\begin{proof}
Since $P_i$ is a polygon, the list $v_i, \dots, v_{i+d-1}$ is a basis in $M$ for all $i$. In particular, there exist $\alpha_{i,j} \in R$ such that

 \[v_{i+d} = \alpha_{0,i}v_i + ... +\alpha_{d-1,i}v_{i+d-1}.\] 
 In other words, for the difference operator
\[\D\ai{_v} : = \alpha_0+  \dots +\alpha_{d-1}T^{i+d-1} - T^{d},\] 
we have $\D\ai{_v} v = 0$. Further, for any $i \in \Z$ we have
$$
M = R\langle v_{i+1} , \dots, v_{i+d} \rangle = R\langle v_{i+1} , \dots, \dots, v_{i+d-1},\alpha_{0,i}v_i \rangle = P_{i+1} \oplus \dots \oplus P_{i+d-1} \oplus R \alpha_{0,i}v_i.
$$
Separately, $$M = P_{i+1} \oplus \dots \oplus P_{i+d-1} \oplus P_i,$$
so, together, we have
$$
P_i = Rv_i =  R  \alpha_{0,i}v_i ,
$$
and thus $\alpha_{0,i}$ is left-invertible and hence a unit for any $i \in \Z$. Therefore, $\D\ai{_v}$ is properly bounded, as desired. 
\end{proof}

The following says that the \ai{correspondence $v \mapsto \D_v$} from Proposition \ref{prop:ri} provides a right inverse for the map  $\bar \varphi :R^\Z[T]_d^{\mathrm{pb}} \to \mathfrak P(\Proj^{d-1}(R))$. Thus, this map is surjective, and the proof of Theorem \ref{DOcorresPOLY} is complete.
\begin{prop}
Assume that $v_i \in M$, where $M$ is a rank $d$ free left $R$-module for stably finite ring $R$, is such that the sequence $Rv_i$ is a polygon. Let $\D \in R^\Z[T]_d^{\mathrm{pb}}$ be a properly bounded degree $d$ difference operator such that $\D v = 0$. Then there is a projective transformation $\Proj((\Ker \D)^*) \to \Proj(M)$ mapping the polygon associated with $\D$ to the polygon $(Rv_i)$. 
\end{prop}
\begin{proof}
Since $\D v = 0$, and the list $v_i, \dots, v_{i+d-1}$ is a basis in $M$ for all $i$, by Proposition \ref{wronskian} we have $\mathrm{Span}\, v = \Ker \D$. This means that the image of the mapping $\psi \colon M^* \to R^\Z$ given by $\xi \mapsto (\xi(v_i))$ is precisely $\Ker \D$. Furthermore, since  $M^*$ \ai{and} $\Ker \D$ are free modules of the same rank, \ai{and $R$ is stably finite,} it follows that $\psi \colon M^* \to \Ker \D$ is a right $R$-module isomorphism.  Consider the dual $\psi^* \colon (\Ker \D)^* \to M$. It is a left $R$-module isomorphism\ai{, and}, for any $\xi \in M^*$, we have
$$
\xi(\psi^*(\eps_{\D, i})) = \eps_{\D, i}(\psi(\xi)) = \xi(v_i),
$$
so $\psi^*(\eps_{\D, i}) = v_i$\ai{. Therefore,} the associated projective map $\Proj((\Ker \D)^*) \to \Proj(M)$ takes the polygon associated with $\D$ to the polygon $Rv_i$, as needed. 
\end{proof}

\section{Integrability of pentagram maps over rings} \label{knownPM}

In  this section, we show that pentagram maps over arbitrary stably finite rings can be viewed as integrable systems. As a first step, in Section \ref{sec:pmdo}, we reformulate the pentagram map in the language of difference operators. This is used in Section \ref{sec:lax} to give a Lax formulation of the pentagram map. It is given in terms of pseudo-difference operators, which are introduced in Section \ref{PDO}.

Throughout this section, the ground ring $R$ is assumed to be stably finite.

\subsection{Pentagram maps in terms of difference operators}\label{sec:pmdo}

By Theorem \ref{DOcorresPOLY}, polygons in the (left) projective plane $\Proj^2(R)$, considered up to projective transformation, are in one-to-one correspondence with elements of the set $R^\Z[T]_3^{\mathrm{pb}}/((R^\Z)^{\times} \times (R^\Z)^{\times})$ of properly bounded left difference operators of degree three, up to left-right action \ai{of $(R^\Z)^{\times}$}. Accordingly, one can describe the pentagram map as a dynamical system on difference operators. In this section we give an intrinsic description of that dynamical system. For technical reasons, we consider the inverse pentagram map, rather than the pentagram map itself. It turns out that that this inverse, phrased in terms  of difference operators, is described by a simple linear equation. In the case when the ground ring is a field, this result was established in \cite[Theorem 1.1]{izopentagramandrefactorization}. 

In order to have well-defined dynamics on difference operators, we only consider polygons  whose image under the inverse pentagram map is a polygon (so that the image corresponds to a difference operator). Denote the inverse pentagram map by $\Psi$, and let $\mathfrak P(\Proj^2(R))^{\mathrm{reg}}$ be the set of projective equivalence classes of polygons in $\Proj^2(R)$ whose image under the inverse pentagram map is a polygon. That way, we have a map
$$
\Psi \colon  \mathfrak P(\Proj^2(R))^{\mathrm{reg}} \to \mathfrak P(\Proj^2(R)).
$$

For a degree three difference operator of the form
\[\D = a+b T +c T^2+d T^3, \quad a,b,c,d \in R^{\Z},\]
 let us define
\begin{equation} \label{eqDplusDminus}
    \D_+ := a + bT \hspace{1cm} \D_- := c T^2+ dT^3.
\end{equation}
\begin{definition} \label{defRelationDO}     Two properly bounded left difference operators of degree three, $\D, \wD \in R^\Z[T]_3^{\mathrm{pb}}$, are said to be \emph{related}, denoted $\D \sim \wD$, if 
     \begin{equation} \label{EqDO}
         \wD_+\D_-=\wD_-\D_+.
     \end{equation}
          \end{definition}
          \begin{prop}
          If $\D \sim \wD$, then $\D_+ (\Ker \D) \subset \Ker \wD$.
          \end{prop}
          \begin{proof}
     By adding $\wD_+ \D_+$ to both sides, the definition \eqref{EqDO} of the relation $\D \sim \wD$  can be rewritten as
         $$
         \wD_+\D =\wD\D_+.
$$
In particular, if $\D \sim \wD$, then, for any $\alpha \in \Ker \D$, we have 
     $
     \wD (\D_+ \alpha) =  \wD_+(\D \alpha) = 0.
     $\qedhere
   \end{proof}
   
   \begin{definition}
     We say that $\D \in  R^\Z[T]_3^{\mathrm{pb}}$ is \emph{regular} if there exists $\wD \in R^\Z[T]_3^{\mathrm{pb}}$  such that $\D \sim \wD$, and the homomorphism of right $R$-modules given by
     $$
     \D_+\mid_{\Ker \D} \colon \Ker \D \to \Ker \wD
     $$
     is surjective (and hence bijective).
     \end{definition}

     When $R$ is the field of real or complex numbers, the set of  $\D \in  R^\Z[T]_3^{\mathrm{pb}}$ satisfying the first condition of this definition (existence of $\wD \sim \D$) is dense by \cite[Lemma 4.10]{izopentagramandrefactorization}. Moreover, surjectivity of $ \D_+\mid_{\Ker \D}$ is, in this case, equivalent to the requirement $\Ker \D_+ \cap \Ker D_- = 0$, which also holds on a dense subset. So, for real or complex numbers, regular operators are dense, hence the name.

\begin{prop} \label{propRelationonQuotientDO}
Suppose that $\D \in  R^\Z[T]_3^{\mathrm{pb}}$ is regular. Then
\begin{enumerate} \item The operator $\wD \in R^\Z[T]_3^{\mathrm{pb}}$ such that $\D \sim \wD$  exists and is unique up to the left action of $(R^\Z)^{\times}$ on $R^\Z[T]_3^{\mathrm{pb}}$. \ai{(Note that, if $\D \sim \wD$, then $\D \sim \alpha\wD$ for any $\alpha \in (R^\Z)^{\times}$. Therefore, the relation $\D \sim \wD$ is invariant under the left action of $(R^{\Z})^\times$ on $\wD$.)}
\item Any operator in the same orbit of the left-right $(R^\Z)^{\times} \times (R^\Z)^{\times}$ action on $R^\Z[T]_3^{\mathrm{pb}}$  as $\D$ is also regular. Moreover, for two regular operators in the same orbit, their related operators are also in the same orbit. 
\end{enumerate}
\end{prop}
\begin{proof}
1. Existence of $\wD$ with the property  $\D \sim \wD$ is a part of the  definition of a regular operator. That $\wD$ also satisfies $\D_+ (\Ker \D) = \Ker \wD$. For any other $\wD' \in R^\Z[T]_3^{\mathrm{pb}}$ such that $\D \sim \wD$, we have $\Ker \wD' \supset \D_+ (\Ker \D) = \Ker \wD$. Furthermore, we know that the mapping $\Ker \wD' \to R^3$, given by taking three consecutive entries of a bi-infinite sequence, is an isomorphism. The same holds for $\wD$. Therefore, the embedding $\Ker \wD \to \Ker \wD'$ must be an isomorphism as well. By Lemma \ref{lemma:inj}, we have $\wD' = \alpha \wD$ for some $\alpha \in (R^\Z)^{\times}$.

2. Consider $\wD$ such that $\D \sim \wD$ and $\D_+ (\Ker \D) = \Ker \wD$. Take an operator $\D' = \alpha \D \beta$ in the same orbit as $\D$. Let $\wD' = \wD \alpha^{-1}$. Then $\D' \sim \wD'$. Furthermore,
$$
\D_+' (\Ker \D') =\alpha \D_+ \beta (\beta^{-1} \Ker \D) = \alpha \D_+ (\Ker \D)   = \alpha \Ker \wD = \Ker \wD'. 
$$
The result follows. \qedhere

\end{proof}

As a result, for a regular operator $\D \in  R^\Z[T]_3^{\mathrm{pb}}$, all its related operators constitute a single orbit of the $(R^\Z)^{\times} \times (R^\Z)^{\times}$ action. Moreover, this orbit remains the same if we replace $\D$ by a left-right-equivalent operator. Denote by $R^\Z[T]_3^{\mathrm{reg}}$ the set of regular operators. Then we have a map
$$
\psi \colon R^\Z[T]_3^{\mathrm{reg}}/((R^\Z)^{\times} \times (R^\Z)^{\times})
 \to R^\Z[T]_3^{\mathrm{pb}}/((R^\Z)^{\times} \times (R^\Z)^{\times})
$$
which takes an orbit of a regular operator  $\D \in  R^\Z[T]_3^{\mathrm{reg}}$ to the orbit of its related operator $\wD \sim \D$. This orbit is independent of the choice of a related operator, by Proposition \ref{propRelationonQuotientDO}.

Recall also that, by Theorem \ref{DOcorresPOLY}, we have a bijection $\varphi \colon  R^\Z[T]_3^{\mathrm{pb}}/((R^\Z)^{\times} \times (R^\Z)^{\times}) \to \mathfrak P(\Proj^2(R))$, the correspondence between difference operators (up to left-right equivalence) and polygons (up to projective equivalence).

\begin{theorem} \label{IPMintermsDO} \begin{enumerate}
  \item   We have\begin{equation}\label{eq:inc}\varphi(R^\Z[T]_3^{\mathrm{reg}}/((R^\Z)^{\times} \times (R^\Z)^{\times})) =\mathfrak P(\Proj^2(R))^{\mathrm{reg}}.\end{equation} In other words, regular difference operators precisely correspond to polygons which are mapped to polygons under the inverse pentagram map.  \item The following diagram commutes.
\begin{equation}\label{cd}
\begin{tikzcd}
R^\Z[T]_3^{\mathrm{reg}}/((R^\Z)^{\times} \times (R^\Z)^{\times}) \arrow[r, "\psi"] \arrow[d, "\varphi"]                    &  \arrow[d, "\varphi"]     R^\Z[T]_3^{\mathrm{pb}}/((R^\Z)^{\times} \times (R^\Z)^{\times})\\
\mathfrak P(\Proj^2(R))^{\mathrm{reg}} \arrow[r, "\Psi"]                                              &   \mathfrak P(\Proj^2(R))   .         
\end{tikzcd}
\end{equation}
In other words, rephrased in terms of difference operators, the (inverse) pentagram map is described by equation \eqref{EqDO}. 
\end{enumerate}
\end{theorem}
\begin{proof}
Suppose that $\D \in R^\Z[T]_3^{\mathrm{reg}}$. We will show that $\Psi(\bar \varphi(\D)) =\varphi(\psi(\D)) $. This will establish the inclusion \ai{$\varphi(R^\Z[T]_3^{\mathrm{reg}}/((R^\Z)^{\times} \times (R^\Z)^{\times})) \subset \mathfrak P(\Proj^2(R))^{\mathrm{reg}}$}, as well as commutativity of the diagram~\eqref{cd}. Let $v_i := \eps_{\D, i} \in (\Ker \D)^*$. Then, by construction of the map $\bar \varphi \colon R^\Z[T]_3^{\mathrm{pb}} \to \mathfrak P(\Proj^2(R))$, the polygon $(Rv_i)$ represents the class of polygons $\bar \varphi(\D)$  associated to $\D$. Furthermore, by Proposition~\ref{prop:spaneps}, we have that $\mathrm{Span}\,v = \Ker \D$. Therefore,
$$
\mathrm{Span}(\D_+ v)  = \D_+ (\mathrm{Span} \, v) = \D_+(\Ker \D) = \Ker \wD,
$$
where $\wD \in R^\Z[T]_3^{\mathrm{pb}}$ is any operator such that $\D \sim \wD$. So, if we set  $\hat v := \D_+ v$, by Proposition~\ref{wronskian} we have that  $\hat v_i, \hat v_{i+1}, \hat v_{i+2}$ is a basis in $(\Ker \wD)^*$ for all $i$. In other words, the points $R\hat v_i$ form a polygon. At the same time, if $\D  =   a+ b T + c T^2+ d T^3$, then
  $$
\hat v_i =  a_i  v_i +  b_i   v_{i+1} = -  c_i  v_{i+2} -  d_i  v_{i+3}.
 $$
Thus,  $\hat v_i \in ( Rv_i \oplus  Rv_{i+1}) \cap ( Rv_{i+2} \oplus  Rv_{i+3})$, meaning that the point $R\hat v_i$  lies on the lines $R v_i \oplus R v_{i+1}$,  $R v_{i+2} \oplus R v_{i+3}$ (those are indeed lines by Lemma \ref{uniqueintpt}) and hence coincides with their intersection point (which is unique also by Lemma \ref{uniqueintpt}). In other words, the points $R\hat v_i$ are the vertices of the image of the polygon $(R  v_i)$ under the inverse pentagram map, which proves that $\Psi(\bar \varphi(\D)) =\varphi(\psi(\D)) $.

Now, we show that $\varphi(R^\Z[T]_3^{\mathrm{reg}}/((R^\Z)^{\times} \times (R^\Z)^{\times})) \supset \mathfrak P(\Proj^2(R))^{\mathrm{reg}}$. Take a polygon $(Rv_i)$, $v_i \in M$ (where $M$ is a rank three free left $R$-module),  such that its image under the inverse pentagram map is a polygon. We need to show that the corresponding difference operator $\D$ satisfying $\D v = 0$  is regular. Arguing as above, we see that the image of the polygon $(Rv_i)$ under the inverse pentagram map is given by the points $R \hat v_i$, where  $\hat v := \D_+ v$. Since this image is a polygon, there exists $\wD \in R^\Z[T]_3^{\mathrm{pb}}$ such that $\wD \hat v = 0$. Let us show that $\D \sim \wD$ and, furthermore, $\D_+ (\Ker \D) = \Ker \wD$. For the latter, note that
$$
\D_+(\Ker \D) = \D_+(\mathrm{Span}\, v) = \mathrm{Span}(\D_+ v) = \mathrm{Span} \,\hat v = \Ker \wD.
$$
As for the former, observe that $\D_+ (\Ker \D) = \Ker \wD$ implies that
$
\Ker \D \subset \Ker \tilde \D$, where \begin{equation}\tilde \D:= \wD_-\D_+ -\wD_+\D_-. \end{equation}
Therefore, $
\Ker \D$ is also contained in the kernel of a degree two operator $\Ker T^{-2}\tilde \D$, which, by Proposition~\ref{prop:kernpb}, implies $T^{-2}\tilde \D = 0$, and hence $\D \sim \wD$. \qedhere

\end{proof}

In \cite{izopentagramandrefactorization}, similar ideas are used to work with the pentagram map itself instead of its inverse. The former is also described by \eqref{EqDO}, but with $\D_\pm$ of the form
\begin{equation}\label{eq:newform}
    \D_+ := a + cT^2, \quad \D_- := b T+ dT^3.
\end{equation}
In that setting, Theorem \ref{IPMintermsDO} does not seem to be true as stated. The argument breaks down where we say that the operator $\tilde \D:= \wD_-\D_+ -\wD_+\D_-$ is of degree two; that is no longer the case for operators of the form \eqref{eq:newform}. 

\subsection{Pseudo-difference operators} \label{PDO}
    A (left) \textit{pseudo-difference operator} over a ring $R$ is a formal Laurent series  in terms of the left shift operator $T$ with coefficients in $R^\Z$. Formally, consider the subring $R^\Z[T,T^{-1}]$ of $\End_R(R^\Z)$ generated by $R^\Z$, $T$, and $T^{-1}$. The ring  $R^\Z[T,T^{-1}]$ consists of operators of the form $\sum \alpha_iT^i$, where $\alpha_i \in R^\Z$, and $\alpha_i \neq 0$ for only finitely many values of $i$. Consider a valuation-type function $\nu \colon R^\Z[T,T^{-1}] \to \Z_+ \cup +\infty$ given by the degree of the lowest order term, where, by definition, $\nu(0) = +\infty$. Then $||\cdot|| := \exp(-\nu(\cdot))$ is a \emph{norm} on $R^\Z[T,T^{-1}]$, in the sense of \cite[Definition 2.11]{Jarden2011}.
    \begin{definition} \label{defLPseudoDO}
  The \emph{ring $R^\Z((T))$ of (left) pseudo-difference operators over a ring $R$} is the completion of $ R^\Z[T,T^{-1}]$ with respect to the norm $||\cdot||$. 
  
\end{definition}
Any pseudo-difference operator $\Q \in R^\Z((T))$ can be uniquely represented as a sum of a series
$
\Q = \sum_{i=m}^{\infty}\alpha_iT^i
$
where $\alpha_i \in R^\Z$, $m\in \Z$, and convergence is relative to the \ai{above} norm.

\begin{prop}
Suppose $\Q = \sum_{i=1}^{\infty}\alpha_iT^i$ is a power series in $T$ without a free term. Then $1 + \Q$ is a unit in $R^\Z((T))$.
\end{prop}
\begin{proof}
We have $ (1 +\Q)^{-1} = \sum_{i=0}^{\infty} (-\Q)^i $. The series converges since the ring $R^\Z((T))$ is complete, and $||\Q|| \leq e^{-1}$.
\end{proof}

\begin{cor}\label{cor:inv}
Let $\Q = \sum_{i=m}^{\infty}\alpha_iT^i$. If $\alpha_m$ is a unit in  $R^\Z$, then $\Q$ \ai{is} a unit in  $R^\Z((T))$.
\end{cor}
\begin{proof}
We have $\Q = \alpha_m T^m (1 + \Rem)$ where $ \Rem$ is a power series in $T$ without a free term. So,~$\Q$ is a product of units and hence a unit. \qedhere
\end{proof}

\begin{remark}

Note that $R^\Z((T))$ is not a subring of $\End_R(R^\Z)$, as the action of a general pseudo-difference operator on a bi-infinite sequence is not well-defined. However, a pseudo-difference operator still defines an endomorphism of the right $R$-module of \emph{eventually vanishing sequences}, i.e., bi-infinite sequences $\alpha_i \in R$ such that $\alpha_i = 0$ for large enough $i$.

\end{remark}

\subsection{Lax representation for pentagram maps over rings}\label{sec:lax}
In this section, we recall the notion of a Lax pair and construct such pairs for pentagram maps over rings. Details on Lax pairs can be found in any textbook on integrable systems, see, e.g., \cite[Section 2.4]{babelon2003introduction}.

Let $X$ be a set, and  $\phi \colon X \dashrightarrow X$ be a dynamical system. We use a dashed arrow to emphasize that the map $\phi$ might not be everywhere defined. For instance, the pentagram map is not defined on the whole set of polygons, since the image of a polygon under that map might not be a polygon. 

Let $S$ be a ring, and $S^\times$ be its group of units.  A  \emph{Lax pair} for $\phi$ is a pair of maps $L \colon X \dashrightarrow S$, $A \colon X \dashrightarrow S^\times$, defined on the domain of $\phi$, such that
$$
L(\phi(x)) = A(x)^{-1}L(x)A(x)
$$
whenever $\phi(x)$ is in the domain of $\phi$. To distinguish meaningful Lax pairs from trivial cases such as $L \equiv 0$, $A \equiv 1$, we say that a Lax pair is \emph{faithful} if the mapping $L$ is injective.

An immediate consequence of the existence of a Lax pair is that functions of the form $f(L(x))$, where $f$ is central (i.e., invariant under conjugation by units) on $S$, are invariants of $\phi$. In many examples, it is possible to show that, for a suitable Poisson bracket on $X$, such invariants constitute a maximal Poisson-commutative family, implying that $\phi$ is an Arnold-Liouville integrable system.

Another framework where Lax pairs play a central role is algebraic integrability. In that setting, $S$ is supposed to be the ring of matrices depending on an additional parameter, known as the \emph{spectral parameter}. In this case, the characteristic polynomial of $\det(L(x) - \lambda \mathbf{1})$ is, for fixed $x$, a function of two variables, $\lambda$ and the spectral parameter. The zero locus of that function is a plane curve, known as the \emph{spectral curve}. The spectral curve is preserved by the dynamical system $\phi$. Moreover, in many cases, one can show that, for a fixed spectral curve, the dynamics of $\phi$ can be identified with linear motion on the Jacobian of the spectral curve, turning $\phi$ into an \emph{algebraic completely integrable system}.

The pentagram map in $\Proj^2(\F)$, where $\F$ is $\R$ or $\C$, is Arnold-Liouville integrable \cite{integrabilityPMOvScTa}; it is also algebraically integrable for any field $\F$ with $\mathrm{char} \, \F \neq 2$ \cite{solovievintegrability, weinreich2023algebraic}. 
A Lax pair with spectral parameter is given, e.g., in \cite[Proposition 4.10]{clusteralgebras}.

In the setting of rings, one cannot expect neither Arnold-Liouville, nor algebraic integrability without further assumptions. Indeed, neither Poisson brackets, nor algebraic curves and their Jacobians make sense in this generality. So,  we take the existence of a Lax pair as a \emph{definition} of integrability, cf. \cite{felipe2015pentagram, higherdimPM2} . In the case of closed polygons, our Lax pair possesses a spectral parameter, see Remark \ref{rem:spectral}. We expect that, under further assumptions, our Lax pair will yield both Arnold-Liouville and algebraic integrability, since over the real and complex numbers it specializes to the standard Lax pair which is known to possess these properties.

The construction of a Lax pair for the pentagram map over a stably finite ring $R$ is as follows. Denote the pentagram map by $\Phi$ and let $\PEC \in \mathfrak P(\Proj^2(R))$ be
 a projective equivalence class of polygons in $\Proj^2(R)$ which is in the domain of $\Phi$ (i.e., the images of polygons in  $\PEC$ under the pentagram map are again polygons). Then $\Phi(\PEC) \in \mathfrak P(\Proj^2(R))^{\mathrm{reg}}$. Let $ \wD \in R^\Z[T]_3^{\mathrm{reg}}$ be an operator associated to $\Phi\left(\PEC\right)$. Then, by Theorem \ref{IPMintermsDO}, the equation  
 \begin{equation}\label{eq:main2}\D_+\wD_-=\D_-\wD_+\end{equation}
 has a solution $\D \in R^\Z[T]_3^{\mathrm{pb}} $, unique up to the left action of $(R^\Z)^\times$,  such that the associated class of polygons coincides with $\PEC$ (note that we have switched the roles of $\D$ and $\wD$). Furthermore, we know that $\D_+$ is a unit in $R^\Z((T))$ by Corollary \ref{cor:inv}. Set \begin{equation}\label{eq:lapair}\mathcal L(\PEC) := \D_+^{{-1}}\D_-, \quad \mathcal A(\PEC) := \wD_+.\end{equation} Note that the definition of $\mathcal L(\PEC)$ is independent on the choice of solution of equation \eqref{eq:main2}. However, both $\mathcal L(\PEC)$ and $\mathcal A(\PEC)$ depend on the choice of $\wD$. Upon replacing $\wD$ by $\alpha \wD \beta$, where $\alpha, \beta \in (R^\Z)^\times$, these operators transform as
 
 \begin{equation}\label{eq:laaction}\mathcal L(\PEC)\mapsto \alpha \mathcal L(\PEC) \alpha^{-1}, \quad
 \mathcal A(\PEC) \mapsto \alpha  \mathcal A(\PEC) \beta.\end{equation}

\begin{prop}\label{prop:lax}
The maps $\mathcal L, \mathcal A$, defined by \eqref{eq:lapair} up to the action \eqref{eq:laaction}, 
constitute a faithful Lax pair for the pentagram map in $\Proj^2(R)$.
\end{prop}
\begin{proof}

From \eqref{eq:main2}, we get
\begin{equation}\label{Laxstep1}
\mathcal L(\PEC) = \D_+^{-1}\D_- = \wD_- \wD_+^{-1}.
\end{equation}
So,
$$
\mathcal A(\PEC)^{-1}\mathcal L(\PEC) \mathcal A(\PEC) = \wD_+^{-1}\wD_- ,
$$
which, up to conjugation action of $(R^\Z)^{\times}$, is equal to $\mathcal L(\Phi(\PEC))$, as needed. To prove that this Lax pair is faithful, we need to show that if 
\begin{equation}\label{Laxstep3}
\D_+^{-1}\D_- = \tilde \D_+^{-1}\tilde\D_-
\end{equation}
for some $\tilde \D \in R^\Z[T]_3^{\mathrm{pb}} $, then $\tilde \D$ represents the same class of polygons $\PEC$ as $\D$. Indeed, from \eqref{Laxstep1}  and \eqref{Laxstep3}, we get
$$
 \tilde \D_+^{-1}\tilde\D_- =  \wD_- \wD_+^{-1},
$$ so
$$\tilde \D_+\wD_-= \tilde\D_-\wD_+.$$
By Theorem~\ref{IPMintermsDO}, the latter means that the class of polygons represented by $\tilde \D$ is the image of $\Phi(\PEC)$ under the inverse pentagram map, i.e., $\PEC$.\qedhere

\end{proof}

We note that, when $R$ is a field, this Lax representation becomes the one given in \cite{clusteralgebras, izopentagramandrefactorization}.

\section{Invariants of pentagram maps on closed polygons} \label{InvariantsGPM}
Here we use the Lax form found in the previous section to give invariants (i.e., conserved quantities, or \emph{first integrals}) of the pentagram map on closed polygons (i.e., polygons \ai{given by periodic sequences $P_i \in \Proj^2(R)$}). In the case of the pentagram map over real numbers, these invariants (in a slightly more general setting of \emph{twisted} polygons) are known to form a maximal Poisson-commuting  family \cite[Theorem 1]{integrabilityPMOvScTa}. We expect that similar results can be obtained for much bigger classes of rings. 

Throughout this section, the ground ring $R$ is assumed to be stably finite.

\subsection{Closed polygons and super-periodic difference operators}
\begin{definition} \label{defSuperPeriodicDO}
A difference/pseudo-difference operator $\D = \sum \alpha_iT^i$ is said to be \emph{periodic}, with period $n$, if all its coefficients $\alpha_i$ are $n$-periodic sequences. Equivalently, $\D$ is $n$-periodic if $\D T^n = T^n \D$. 
    An $n$-periodic difference operator is said to be \textit{super-periodic} if its kernel consists entirely of $n$-periodic sequences\ai{, cf. \cite{krichever2015commuting}}.
\end{definition}
This class of difference operators is naturally acted upon the left and right by the subgroup of periodic bi-infinite sequences of units.
\begin{prop} \label{propSuperDOperiodic}
    Let $\D$ be a super-periodic left properly bounded difference operator with period $n$. Consider also $\alpha, \beta \in (R^\Z)^{\times}$. Then, $\alpha \D \beta $ is super-periodic with period $n$ if and only if $\alpha$ and $\beta$ are $n$-periodic. 
\end{prop}
\begin{proof}
    Suppose first that $\alpha$ and $\beta$ are $n$-periodic. Then
    $$
     \Ker \alpha \D \beta  = \beta^{-1}\Ker \D,
    $$
    so all elements of $\Ker \alpha \D \beta $ are $n$-periodic. 
    
    Conversely, suppose that that $\alpha\D \beta$ is super-periodic with period $n$. Fix $i \in \Z$, and choose $a \in  \Ker \alpha\D \beta $ such that $a_i = 1$. Since $ \alpha\D \beta$ is super-periodic with period $n$, we also have $a_{i+n} = 1$. Further, since $a \in  \Ker \alpha\D \beta $, we have $\beta  a \in \K$. Since $\D$ is super-periodic with period $n$, this gives $\beta_i a_i = \beta_{i+n} a_{i+n}$, so, $\beta_{i+n} = \beta_i$. Additionally, given that $i$ was arbitrary, this means that $\beta$ is $n$-periodic. But then $\alpha$ is a product of $n$-periodic pseudo-difference operators:
  $$
  \alpha =  \alpha\D \beta (\D \beta)^{-1},
  $$
  and hence $n$-periodic. 
\end{proof}
This gives the following periodic version of Theorem \ref{DOcorresPOLY}.
\begin{cor} \label{corDOclosed}
Let $R$ be a stably finite ring. Then, there is a natural one-to-one correspondence between the following sets:
\begin{enumerate}
    \item Projective equivalence classes of closed $n$-gons in the left  projective space $\Proj^{d-1}(R)$.
    \item Properly bounded super-periodic left difference operators of degree $d$ and period $n$ with coefficients in $R$, up to multiplication on the left and right by $n$-periodic sequences of units in $R$.
\end{enumerate}
\end{cor}
\begin{proof} 
If $\D$ is a super-periodic operator with period $n$, then $\eps_{\D}$ is an $n$-periodic sequence in $(\Ker \D)^*$. So, the mapping $\bar \phi \colon R^\Z[T]_d^{\mathrm{pb}} \to \mathfrak P(\Proj^{d-1}(R))$ takes super-periodic operator\ai{s} with period $n$ to closed polygons with the same period. Conversely, given a closed $n$-gon $(P_i)$, we can choose spanning vectors $v_i \in P_i$ so that $v_{i+n} = v_i$, and then the operator $\D$ such that $\D v = 0$ can be taken to be $n$-periodic. So, the set of properly bounded super-periodic left difference operators of degree $d$ and period $n$ \ai{is mapped onto} the set  of projective equivalence classes of closed $n$-gons . Furthermore, by Proposition \ref{propSuperDOperiodic} preimages of every class of polygons under this mapping are precisely the orbits of the left-right action of $n$-periodic sequences of units.\qedhere

\end{proof}
One also has a periodic version of Theorem \ref{IPMintermsDO}. In particular, the diagram \eqref{cd} still commutes, assuming that we replace all objects with their periodic counterparts. As a result, to any projective equivalence class of closed $n$-gons we can assign an $n$-periodic pseudo-difference operator $\Q := \D_+^{-1}\D_-$, defined up to conjugation by $n$-periodic sequences of units. Upon application of the pentagram map, the operator $\Q$ changes to one that is conjugate to $\Q$ within the ring of $n$-periodic pseudo-difference operators.

\subsection{The trace of a periodic pseudo-difference operator} \label{invariantInnerProduct}

For a ring $R$, consider the additive subgroup $[R,R] \subset R$ generated over $\Z$ by all commutators $[a,b]:=ab - ba$. The Abelian group $R / [R,R]$ is known as the \emph{cyclic space} of $R$, and the quotient map $R \to R/[R,R]$ as the \emph{universal trace}.

\begin{definition} \label{defTrace}
    For an $n$-periodic left pseudo-difference operator  $\Q = \sum \alpha_iT^i \in R^\Z((T))$, its trace $\Tr\, \Q$ is an element of the cyclic space $R / [R,R]$ defined by
$$
\Tr\, \Q := \sum_{i=1}^n \alpha_{0,i},
$$
where the right-hand side is understood modulo $[R,R]$.
\end{definition}
Thus, the trace is a homomorphism of Abelian groups $\Tr \colon R^\Z((T)) \to R / [R,R]$. The following shows it is indeed a trace (that is, it satisfies the cyclic property).

\begin{lemma}\label{trinv}
   Let $\Q$ and $\wQ$ be two $n$-periodic left pseudo-difference operators. Then 
    $$
    \Tr\, \Q\wQ = \Tr\, \wQ\Q.
    $$

\end{lemma}
\begin{proof}
We need to show that  $\Tr ( \Q\wQ - \wQ\Q) = 0$. Note that the left-hand side is bilinear over~$\Z$. Furthermore, if 
we endow the cyclic space with discrete topology, then the trace function becomes continuous, so the operation $\Q, \wQ \mapsto \Tr ( \Q\wQ - \wQ\Q)$ distributes over infinite sums. Hence, it suffices to verify the statement in the case $\Q = \alpha T^k$, $\wQ = \beta T^{\ai{l}}$, where $\alpha, \beta \in R^\Z$ are $n$-periodic. Clearly, $\Tr \, \Q\wQ =\Tr \, \wQ\Q= 0$ if $ l \neq -k$, so suppose $l = -k$. Then
$$
\Tr \, \Q\wQ = \Tr\, \alpha T^k \beta T^{-k} = \sum_{i=1}^n \alpha_{i}\beta_{i+k}.
$$
while
$$
\Tr \, \wQ\Q = \Tr\, \beta T^{-k} \alpha T^k  = \sum_{i=1}^n \beta_{i}\alpha_{i-k}.
$$
Modulo $[R,R]$, these two expressions are equal. \qedhere
\end{proof}

\subsection{Invariants of the pentagram map}

By the above, to any projective equivalence class $\PEC$ of closed $n$-gons we can assign an $n$-periodic pseudo-difference operator $\mathcal L(\PEC) := \D_+^{-1}\D_-$, defined up to conjugation by $n$-periodic bi-infinite sequences of units. Upon application of the pentagram map, the operator $\mathcal L(\PEC) $ changes to a conjugate one. Now, consider the functions on the space $\mathfrak P(\Proj^2(R))$ of projective equivalence classes of polygons defined by
\begin{equation} \label{eqInvariantFu}
    f_{ij}(\PEC) := \Tr(T^{in}\mathcal{L}(\PEC)^j), \hspace{1cm} i \in \Z, j \in  \Z_{\ge 0},
\end{equation}
where $\mathcal L$ is defined as described above.

\begin{prop}
    \begin{enumerate}
        \item The functions $f_{ij}$ are independent of the choice of $\mathcal L(\PEC)$ and, hence, are well-defined functions on  $\mathfrak P(\Proj^2(R))$ valued in the cyclic space $R/[R,R]$.
        \item The functions $f_{ij}$ are invariant under the pentagram map. 
        \end{enumerate}    
\end{prop}
\begin{proof}
Since $\mathcal L(\PEC)$ is defined up to conjugation, and the action of the pentagram map is also given by conjugation, it suffices to show that the functions $\Q \mapsto  \Tr(T^{in}\mathcal{\Q}^j)$ on $n$-periodic pseudo-difference operators are conjugation-invariant ({central}). Indeed,
$$
 \Tr(T^{in}(\wQ \Q \wQ^{-1})^j) =  \Tr(T^{in}\wQ \Q^j \wQ^{-1}) =  \Tr(\wQ T^{in} \Q^j \wQ^{-1}) =  \Tr(T^{in}\Q^j),
$$
where in the second equality we used periodicity of $\wQ$, and in the last equality Lemma \ref{trinv}. 
\end{proof}

\begin{example}
If $R = \Mat_k(\F)$, then $R/[R,R] \simeq \F$, so the invariants $f_{ij}$ are valued in the ground field $\F$. 
\end{example}

\begin{example}
The ring $R = \DN$ of dual numbers is commutative, so  $[R,R] = 0$, and the invariants $f_{ij}$ are valued in the ring $\DN$ itself. The real parts of these invariants can be understood as follows.  There is a natural projection $\pi \colon \Proj^2(\DN) \to \Proj^2(\R)$ which sends every affine line to its direction. The skewer pentagram map and the ordinary pentagram map are semi-conjugate via this projection (or, more precisely, via the induced projection on polygon spaces). In particular, $\pi$-lifts of the invariants of the usual pentagram map are invariants of the skewer map. These $\pi$-lifts are precisely the real parts of $f_{ij}$ (this easily follows from the fact that the real part function $\Re \colon \DN \to \R$ is a homomorphism of rings and hence takes all objects associated with the skewer map to the corresponding objects for the pentagram map over the reals). The non-real parts of $f_{ij}$ are additional invariants which do not come from the pentagram map in $\Proj^2(\R)$.
\end{example}

\begin{remark}\label{rem:spectral}
The ring of $n$-periodic pseudo-difference operators over $R$ is isomorphic to the ring $\Mat_n(R)((z))$ of formal Laurent series over $\Mat_n(R)$, cf. \cite[Remark 3.8]{izopentagramandrefactorization}. The isomorphism is given by
$$
\sum_{i=m}^{\infty}\alpha_iT^i \mapsto \sum_{i=m}^{\infty}\mathrm{diag}(\alpha_{i1}, \dots, \alpha_{in})\Lambda^i, 
$$
where $\mathrm{diag}(\alpha_{i1}, \dots, \alpha_{in})$ is the diagonal matrix with entries $\alpha_{i1}, \dots, \alpha_{in}$, and
$$
\Lambda = zE_{n,1}+\sum_{i=1}^{n-1} E_{i,i+1}.
$$
Here $E_{i,j} \in \Mat_n(R)$ is the matrix with a $1$ at position $(i, j)$ and zeros elsewhere. Therefore, in the periodic case, the Lax representation given by Proposition \ref{prop:lax} can be thought of as being valued in $\Mat_n(R)((z))$. In other words, it is a \emph{Lax representation with spectral parameter}. Also note that, for $R = \Mat_k(\F)$, we have $\Mat_n(R) \simeq \Mat_{nk}(\F)$. So, in this case, one gets Lax matrices whose entries are in a field, allowing one to study the corresponding (Grassmann) pentagram maps using standard tools of integrable systems theory. This Lax representation is different from and in a sense dual to the one given in \cite{felipe2015pentagram}.
\end{remark}


\bibliographystyle{plain}
\addcontentsline{toc}{section}{References}
\bibliography{ref.bib}

\end{document}